\documentclass{article}
\usepackage{amsfonts,amsthm}
\usepackage{graphicx,amsmath,amssymb,mathrsfs}
\usepackage{hyperref}

\def\cvd{~\vbox{\hrule\hbox{%
     \vrule height1.3ex\hskip0.8ex\vrule}\hrule } }

\newtheorem{theorem}{Theorem}[section]
\newtheorem{lemma}[theorem]{Lemma}
\newtheorem{proposition}[theorem]{Proposition}
\newtheorem{corollary}[theorem]{Corollary}
\newtheorem{remark}[theorem]{Remark}
\newtheorem{example}[theorem]{Example}
\newenvironment{keywords}
{\noindent {\bf Keywords.}}{}
\newenvironment{AMS}
{\noindent {\bf AMS subject classifications.}}{}

\newcommand{\ct}[1]{#1^{\mbox{\tiny\sf T}}}            % non-offensive transpose symbol
\newcommand{\cj}[1]{#1^{\mbox{\tiny\sf G}}}            % non-offensive G-transpose symbol
\newcommand{\R}{{\mathbb R}}                           % blackboard R for real
\newcommand{\BC}{{\mathbb C}}                          % blackboard C for complex
\newcommand{\BH}{{\mathbb H}}                          % blackboard H for Hilbert space
\newcommand{\cC}{\mathscr{C}}                          % nice script C
\newcommand{\cL}{\mathscr{L}}                          % nice script L
\newcommand{\cH}{\mathscr{H}}                          % nice script H
\newcommand{\wek}[1]{\textsf{$\textbf{#1}$}}           % bold sanserif font for vectors
\newcommand{\bm}[1]{\mbox{\boldmath $#1$}}             % boldmath for vector Greek and for the 'ell' l

\newcommand{\expwek}[2]{e^{#1\,\mbox{\footnotesize $\wek #2$}}}     % scalar and vector in the exponent
\newcommand{\diag}{{\rm diag}}
\newcommand{\be}{\begin{equation}}
\newcommand{\ee}{\end{equation}}
\newcommand{\bmat}{\begin{bmatrix}}
\newcommand{\emat}{\end{bmatrix}}
\newcommand{\set}[1]{\left\{\,#1\,\right\} }                  % set {...}
\newcommand{\jor}{\Lambda}
\newcommand{\tZ}{\tilde{Z}}
\newcommand{\tF}{\tilde{F}}

\newcommand{\cc}{\overline{\wek c}}
\newcommand{\cg}{\overline{\bm \gamma}}
\newcommand{\lr}[2]{\left\langle\, #1\,,\,#2\,\right\rangle}  % nice triangular brackets <...,...>
\newcommand{\kor}{\,\leftrightarrow\,}                        % used for Pauli coding
\newcommand{\df}{\,\stackrel{\mbox{\footnotesize\sf def}}{=}\,}
\newcommand{\imp}{\,\Rightarrow\,}
\newcommand{\tss}[1]{\textsuperscript{#1}}
\newcommand{\vspandexsmall}{\rule[-11pt]{0pt}{22pt}} % extra small vertical space in arrays
\newcommand{\lsig}{\,^{\sigma}\! }                     % left superscript sigma
\newcommand{\cross}{\!\times\!}                      % tighter cross product
\newcommand{\AND}{\quad\mbox{and}\quad}              % Roman 'and' in display math with extra spaces
\newcommand{\vv}{\vspace{10pt}                       % extra vertical space

}
\newcommand{\Sum}{\displaystyle\sum }                % forced displaystyle
\newcommand{\Frac}{\displaystyle\frac}               % forced displaystyle
\newcommand\sign[1]{{\rm sign}(#1)}                  % forced Roman
\newcommand{\refrm}[1]{{\rm(\ref{#1})}}              % consistent Roman reference even within italics
\newcommand{\tr}{{\rm tr}\,}                         % forced Roman
\begin{document}

\bibliographystyle{plain}
\title{
Lorentz transformation from an elementary point of view\footnote{\href{http://repository.uwyo.edu/ela/vol31/iss1/56/}{Electronic Journal of Linear Algebra, Volume 31, pp. 794-833.}}}

\author{
Arkadiusz Jadczyk\thanks{Laboratoire de Physique Th\'{e}orique, Universit\'{e} de Toulouse III \& CNRS,  118 route de Narbonne, 31062 Toulouse, France,
              {ajadczyk@physics.org}}
\quad and\quad
Jerzy Szulga\thanks{
              Department of Mathematics and Statistics, Auburn University,  Auburn, AL 36849, U.S.A.,
              Tel.: 1-334-844-6598, Fax: 1-334-844-6555,
              {szulgje@auburn.edu}.}}

\thispagestyle{empty}
% Authors and running title to go on top of each page
\pagestyle{myheadings}
\markboth{A. Jadczyk and J. Szulga}{Lorentz transformation from an elementary point of view}
\maketitle

\begin{abstract}
Elementary methods are used to examine some nontrivial mathematical issues underpinning the Lorentz transformation. Its eigen-system is characterized through the exponential of a $G$-skew symmetric matrix, underlining its unconnectedness at one of its extremes (the hyper-singular case). A different yet equivalent angle is presented through Pauli coding which reveals the connection between the hyper-singular case and the shear map.
\end{abstract}
\vspace{5pt}

\vspace{5pt}


%%%%%%%%%%%%%%%%%%%%%%%%%%%%%%%%%%%%%%%%%%%%%%%%%%%%%%%%%%%%%

\section{Introduction}\label{Intro}
The studies of Lorentz transformation usually place it within quite sophisticated theories. Yet, it
stands to reason to believe that a significant part of the theory could and should be derived
within the standard framework of the linear algebra and Euclidean geometry, that is, ``simple from
simple''. A reader can see a similar approach in other publications such as \cite{Ludyk2013} (cf.
Subsection \ref{ss:Max}). Our title  paraphrases the title of the famous book by Hans Rademacher
\cite{rademacher1983}.

We simplify and expand our previous work \cite{arja-jesz2014} (a comment on \cite{ozdemir2014}).
Our propositions admit simple proofs that avoid specialized  terminology or advanced  theories. In
other words, we confine to the ``lingua franca'' of the common mathematical background so that
standard textbooks on linear algebra and calculus suffice as a technical reference. Still, we
couldn't resist ourselves from introducing  a few linguistic inventions, one is the notion of a
``slider'' for a specific group of operators (compare \cite{bartocci2014}), and another well known
action we dub a ``jaws operator'', for  these coined names capture the essence of the underlying
properties. We pay a special attention to ``hyper-singular'' matrices whose mere singularity is
enhanced by their extreme behavior.

A more advanced setup from the point of view of mathematics or physics, based on tensors, Lie
groups and algebras, etc.,  can be found in numerous and extensive literature, cf., e.g.
\cite{zenirod1992,barutzenilaufer1994,gallier2005,naber2012,minguzzi2013}, to name but a small
sample. Exponential formulas similar to one in Proposition \ref{prop:exp} below have been
constantly derived by various methods while they could be traced back to 1960s (cf.\
\cite{dimimlad2005} and References therein; I.\ M.\ Mladenov, personal communication, September 28,
2016).

The scope of the paper is as follows.

Section \ref{sect:Gskew}: The matrix $G=\diag(1,-1,-1,-1)$ entails  the  $G$-transformation,
$X\mapsto G\ct X G$ and the notions of  $G$ symmetric, $G$-skew symmetric, and $G$-orthogonal
matrices. We examine in detail the eigen-system of a $G$-skew symmetric matrix and show that it is
unconnected at the extreme of hyper-singular matrices.

Section \ref{sect:Lorentz}: An eigen-analysis of the Lorentz transformation yields  its intrinsic
structure. We establish the surjectivity of the exponential map by elementary means. The fact that
every Lorentz matrix is an exponential of a Maxwell matrix is well known but its derivation is
typically quite involved (cf., e.g., \cite{barutzenilaufer1994}).

Section \ref{sect:Pauli coding}: With the help of the Pauli's axiomatic formalism, we represent
Lorentz matrices as so called ``jaws operators'', which again leads to the exponentiation of
$G$-skew symmetric matrices. The defective Lorentz matrices, usually omitted in the literature,
appear as images of the shear transformations.

\subsection{Notation and basics}\label{s:not}
 We consider complex vectors and matrices in a finite dimensional real or complex Euclidean space. We use the usual  $n\times 1$ matrix to represent a vector, indicated by the boldface sanserif font, e.g., $\wek x$. The transpose is denoted by $\ct X$ and $X^*$ marks the Hermitian transpose. For the scalar product we write $(xy)=\ct{\wek x}\wek y$ and the length is written as $x=||\wek x||$. We use the universal symbols $I$ and $O$ for the identity and null matrix whose dimension will be clear from the context.  Two matrices $X$ and $Y$ are called {\bf product  orthogonal}, if $XY=YX=0$, while they are called {\bf orthogonal} if
 \[
 \lr{X}{Y}\df \frac{1}{n}\, \tr X^*Y=0.
 \]
 Among a plethora of equivalent norms we choose the normalized Frobenius (a.k.a. Hilbert-Schmidt) norm
\[
||X||=\frac{1}{\sqrt{n}}\,\sqrt{\tr X^*X}, \quad \mbox{i.e., }||X||^2=\frac{1}{n}\sum_j\sum_k |a_{jk}|^2=\frac{1}{n}\sum_k s_k^2,
\]
where $s_k$ are nonnegative singular values of $X$, i.e., $s_k^2$ are eigenvalues of the symmetric
nonnegative definite matrix $X^* X$.  Hence
\[
||I||=1\quad\mbox{and}\quad  ||XY||\le ||X||\cdot ||Y||,
\]
which yields the well-defined exponential $e^X=\sum_{n=0}X^n/n!$ so that
\[
\left\|e^X\right\|\le e^{||X||}, \qquad  \left\|e^X-e^Y\right\|\le e^{\max(||X||,||Y||)}\,||X-Y||,
\]
and
\[
 \left\|e^X-I-X\right\|\le \frac{1}{2}\,||X||^2\,e^{||X||}.
\]
Hence, if a matrix $X=X(t)$ is continuous with respect to a complex or real variable $t$, so is
$E(t)=e^{X(t)}$. Let  $X(t)$ be differentiable at $0$ with $F=X'(0)$ and $X(0)=I$. Then $X=I+tF+R$,
where the remainder is simply  $R=X-I-tF$  and $||R||/t\to 0$ as $t\to 0$. Hence the above
inequalities entail the estimate
\[
\begin{array}{c}
\Frac{\left\|e^{I+tF+R}-I-tF\right\|}{t}\le \Frac{ \left\|e^{I+tF+R}-e^{I+tF} \right\|}{t}
+ \Frac{\left\|e^{I+tF}-I-tF\right\|}{t}\\
\le c\,\left(\Frac{||R||}{t}+t\,||F||^2\right)\to 0,
\end{array}
\]
for some constant $c$, independent of $t$.  In other words, $E(t)$ is differentiable at $0$ and
$E'(0)=F$. If, additionally, $X(t)$ has the property \be\label{group} X(s)X(t)=X(s+t), \ee so
$X(t)$ commute, then $F$ commutes with $X(t)$ and derivatives of arbitrary orders exist at every
point. Moreover, \be\label{expass} \frac{d^n X(t)}{dt^n}=F^n \,E(t),\qquad E(t)=e^{t\,F}. \ee The
matrix $F$ is called the {\bf generator} of the matrix-valued process $X(t)$.

The spectral representation of the Hermitian square, $X^*X=UD^2 U^*$, where $U$ is a unitary matrix
and $D=\diag (s_k)$, yields the square root $P=\sqrt{X^*X}\df UDU^*$, i.e. $P^2=X^*X$. A
nonsingular $X$ (so $P$ is also nonsingular) admits the unique polar representation $X=UP$, with
the unitary matrix $U=XP^{-1}$.

\subsubsection{Rodrigues formula}\label{ss:Rodrigues}
The matrix calculus, involving analytic functions $f(A)$ with matrix argument, has been of great
interest (cf., e.g.,  \cite{ozdemir2014}, \cite{arja-jesz2014}, and references therein)  since the
origins of linear algebra, with the exponential $e^Z$ playing the leading role.

 While the diagonalization is of utmost importance, the pattern of powers $Z^n$ is also extremely useful. For example, for a $c$-idempotent matrix $Z$ (i.e., $Z^2=cZ$, so for $c=1$ the matrix $Z$ represents a projection, while Z is nilpotent of order 2 for $c=0$) one obtains the formula
\[
f(tZ)=\left\{
\begin{array}{ll}
f(0)\,I+\Frac{f(ct)-f(0)}{c}\,Z, &\mbox{ for $c\neq 0$},\\
f(0)\,I+f'(0)\,Z,& \mbox{ for $c=0$},
\end{array}\right.
\]
{valid for $|t|<r$ for some $r\in [0,\infty]$}. Similar albeit more complicated formulas result
from more general properties. Suppose that $Z$ is $(c,k)$-idempotent, i.e., $Z^k =c\,Z$ for a
scalar $c=c_{k}$, where $k$ is the smallest integer exponent with this property.  Then a neat
pattern of powers $Z^n$  entails a more sophisticated version of the above formula. For example, if
$k=3$ and $c=\pm \,a^2$, then
\begin{equation}\label{eZ}
e^{tZ}=
\left\{\begin{array}{lll}
\vspandexsmall I +\Frac{\sin at}{a}\, Z+\Frac{1-\cos at}{a^2}\,Z^2 &\mbox{when $c<0$} &\quad\mbox{(a)}\\
\vspandexsmall I +\Frac{\sinh at}{a}\, Z+\Frac{\cosh at-1}{a^2}\,Z^2 &\mbox{when $c>0$} &\quad\mbox{(b)}\\
I + tZ+\Frac{t^2}{2}\,Z^2 &\mbox{when $c=0$}&\quad\mbox{(c)}\\
\end{array}\right..
\end{equation}
\begin{example}\label{ex:Rod} \rm
A vector $\wek h$ generates a singular skew symmetric matrix
\[
V=V(\wek h)\df
\left[
\begin{array}{rrr}
0&h_3&-h_2\\
-h_3&0&h_1\\
h_2&-h_1&0\\
\end{array}
\right].
\]
In other words, $V(\wek h)\wek u=\wek u\times \wek h$. Let $\wek h$ be a unit vector and $\theta$
be a real scalar. Then  $V^3=-V$, so
\[
e^{-\theta\,V(\wek h)}=I-\sin\theta\, V(\wek h)+(1-\cos\theta)\,V^2(\wek h),
\]
which is the classical  Rodrigues formula for a matrix representation of the positive (according to
the Right Hand Rule) rotation by an angle $\theta$ about the unit vector $\wek h$. Conversely, if
$R$ is orthogonal with $\det (R)=1$, then $R$ is a rotation by $\theta$ about an axis $\wek h$, so
$R=e^{-\theta V(\wek h)}$.

We note the basics properties of the cross product
\begin{equation}\label{eq:cross}
V(\wek h)\wek d+V(\wek d)\wek h=\wek 0,
 \qquad V(\wek h)V(\wek d)=\wek {h} \ct {\wek {d}}-(dh)\, I,
\end{equation}
The three matrices $I,\, V,\, V^2$ form a basis of the span of powers of $V$. Since $V^2(\wek
h)=\wek h\ct{\wek h}-I$, hence  $I,\, V,\, \wek h\ct{\wek h}$ is a basis as well. Therefore
\[
e^{-\theta\,V(\wek h)}\wek u=\cos\theta\, \wek u-\sin\theta\, \wek u\times \wek h+(1-\cos\theta)\,(uh)\,\wek h.
\]
\end{example}
\subsubsection{Maxwell equations and Lorentz transformation}\label{ss:Max}
The matrix
\[
G\df\diag(1,-1,-1,-1)=\bmat 1& 0\\ 0 &-I \emat
\]
entails the {\bf $G$-transpose} $\cj X\df G\ct XG$. That is,
 \[
 \cj{\bmat c& \ct{\wek y}\\ \wek x  &S \emat}=\bmat  c&-\ct{\wek x}\\ -\wek y &\ct S \emat.
 \]
Clearly, $\cj {(XY)}=\cj Y\cj X$, and we arrive at the notions of a $G$-symmetric ($\cj X=X$) or
$G$-skew symmetric ($\cj X=-X$) matrix. Every $G$-skew symmetric $4\times 4$ matrix has the null
diagonal form \be\label{Gskew} F=F(\wek d,\wek h)\df\bmat 0 & \ct{\wek d}\\ \wek d & V(\wek
h)\emat, \quad\mbox{where}\quad \ct V=-V, \ee
 so we may perceive $(\wek d,\wek h) \mapsto F$  as a linear mapping from $(\R^3)^2$ into $\R^{4\times 4}$.

Now, let us examine the pattern of Maxwell equations. For a $\R^3\mapsto \R^3$ vector field $\wek
F$ we use the formal notation ${\rm div}\,\wek F=\nabla \cdot\wek F$ and $\wek{curl}\, \wek
F=\nabla\times\wek F$.  We say that two vectors fields $\wek B$ and $\wek E$ of the variable
$(t;\wek x)$ satisfy {\bf Maxwell's equations} (cf.  Gottlieb \cite{gottlieb2004}) with  a scalar
$\rho$ and a vector field $\wek J$ such that
\begin{align}
\nabla\times \wek E+\partial_t\wek B=0\label{maxw1},\\
\nabla\times \wek B-\partial_t\wek E=\wek J\label{maxw2},\\
\nabla\cdot\wek E=\rho\label{maxw3},\\
\nabla\cdot\wek B=0\label{maxw4}.
\end{align}
Using the Reverse Polish Notation (where the operator follows the operand) we rewrite Maxwell's
equations in the matrix form, using the aforementioned operator $F(\wek E,\wek B)$ and  its
$G$-conjugate $\tF=F(\wek B,-\wek E)$:
\[
\begin{array}{rccl}
\vspandexsmall
\mbox{\refrm{maxw2} \& \refrm{maxw3}:}\quad &
\bmat 0 & \ct{\wek E}\\ \wek E& V(\wek B)\emat \!\!&\!\! \bmat-\partial_t\\ \nabla\emat
&=\bmat \rho\\ \wek J\emat,\\
\vspandexsmall
\mbox{\refrm{maxw1} \& \refrm{maxw4}:}\quad &
\bmat 0 & \ct{\wek B}\\ \wek B& -V(\wek E)\emat \!\!&\!\! \bmat-\partial_t\\ \nabla\emat   &=
\bmat 0\\ \wek 0\emat.
\end{array}
\]
The fundamental invariance of Maxwell's equations  under the Lorentz transformation is well known
(cf., e.g.,  \cite[Sect.1.8]{Ludyk2013} for an elementary yet rigorous justification). \vv
Therefore, we call a nonzero matrix  \refrm{Gskew} a {\bf Maxwell matrix}. On the other hand,  by a
{\bf Lorentz matrix} we understand a  $G$-orthogonal matrix $X$, i.e., $X^{-1}=\cj X$, and we call
it {\bf proper}, if $\det X=1$ and $X_{00}>0$ (the rows and columns are enumerated $0,1,2,3$ to
distinguish the temporal variable from the three spatial variables). \eject
\section{Eigen-system of $G$-skew symmetry}\label{sect:Gskew}
\subsection{Eigenvalues and skew conjugate}\label{s:eigenvalues}
\begin{proposition}
Let $F$ be a $4\times 4$ Maxwell matrix.
\begin{itemize}
\item[{\rm (a)}] The eigenvalues are $\pm \sigma,\, \pm i\,\theta$  for some real parameters
    $\sigma,\,\theta$, which satisfy an elementary system of quadratic equations, solvable at
    will:
\begin{equation}\label{eqs}
\sigma^2-\theta^2=d^2-h^2,\qquad \sigma^2\,\theta^2=(dh)^2.
\end{equation}
\item [{\rm (b)}] \label{eq:||} Additionally, $(h^2-\theta^2)(h^2+\sigma^2)=h^2d^2-(dh)^2=||\wek
    d\times \wek h||^2\ge 0$. In particular, $\sigma^2\le d^2$ and $\theta^2\le h^2$, and either
    equality occurs iff $\wek d||\wek h$.
\item [{\rm (c)}] $F$ is diagonalizable iff at least one of the eigenvalues is nonzero. The
    quadruple zero eigenvalue occurs if and only if $ \wek d\perp \wek h\quad\mbox{and}\quad
    d=h$, in which case we call $F$ {\bf hyper-singular}.
\end{itemize}
\end{proposition}
\begin{proof}
Since $F$ is {$G$-skew symmetric}, the eigenvalue pattern follows directly from the determinant,
computable by routine:
\begin{equation}\label{detF}
\det(F)=-(dh) ^2
\end{equation}
Since the double eigenvalues of $F^2$ are $\sigma^2$ and $-\theta^2$, so
\[
det(F^2)=\sigma^4\,\theta^4,\qquad \tr(F^2)=2(\sigma^2-\theta^2).
\]
On the other hand, a little calculation shows that $\tr(F^2)=2(d^2-h^2)$ and with \refrm{detF} in
mind we arrive at the stated equations.

Equations \refrm{eqs} entail the formula in part (b), which in turn yields the corollaries.

 Regarding part (c), a nonsingular Maxwell matrix with four distinct eigenvalues is clearly diagonalizable. So it is  when only one parameter is 0, due to the rank at least 2. Finally, a nonzero matrix with the quadruple zero eigenvalue does not posses an eigenbasis.
\end{proof}
\begin{remark}\label{rem:choose}
{\rm We have $(dh)=\pm \sigma\theta$. The sign issue was present and resolved in Example
\ref{ex:Rod}. In Rodrigues formula the negative sign in the exponent has ensured the right
orientation of the 3-space, which has agreed with the order of factors in the cross product
operator $V(\wek h)\wek u=\wek u\cross\wek h$. Of course, the inverse order could be chosen as well
but then it would affect the Rodrigues formula and possibly most of the formulas in this paper.
Once the choice is made, all sign set-ups must be consistent with it. Therefore, to prevent the
unnecessary ambiguity that may arise in the case of nonzero eigenvalues, when associating sigma and
theta with a given F we henceforth will assume  that $\sigma\ge 0$ and choose the symbol $\theta$
for the eigenvalue $\pm\theta i$ for which $(dh)=\sigma\theta$ (see Example \ref{ex:1} below).}
\end{remark}

We call $F$ {\bf normalized} if $\theta^2+\sigma^2=1$, which can be assumed w.l.o.g. for most of
these notes. For any function $\Phi$ on $(\R^3)^2$, we define the skew conjugate
\[
\tilde{\Phi}(\wek u,\wek v)=\Phi(\wek v,-\wek u).
\]
In particular, we obtain a $G$-skew symmetric \be\label{Gskew con} \tF=F(\wek h,-\wek d)=\bmat 0 &
\ct{\wek h}\\ \wek h  & -V(\wek d) \emat \quad\imp\quad \tilde{\tilde{F}}=-F. \ee From the
properties \refrm{eq:cross} of the cross product we infer that
\begin{equation}\label{dual}
\begin{array}{lccl}
\mbox{(a) }\quad & F\tF =\tF F=(dh)\, I         &=&\sigma\theta\,I,    \\
\mbox{(b) }\quad & F^2-\tF^2=(d^2-h^2)\,I       &=&(\sigma^2-\theta^2)\,I,    \\
\mbox{(c) }\quad & F^3=(d^2-h^2)\,F+(dh)\,\tF   &=&(\sigma^2-\theta^2)\,F+\sigma\theta\,\tF.\\
\end{array}
\end{equation}
Hence, if $\wek d\perp\wek h$, the algebra generated by $F$, i.e., the vector space spanned by the
powers $F^n$, is 3-dimensional with a basis $I, \,F,\,F^2$, while for nonorthogonal vectors the
span is 4-dimensional with a basis $I, \,F,\,F^2,\, \tF$.
\begin{example}\label{ex:1}
\rm Let $\wek d=\ct{\bmat 1&1&1\emat},\, \wek h=\ct{\bmat 1&-2&-1\emat}$. Thus $(dh)=-2$ and
\[
F=\left[\begin{array}{c|ccc}
0      &  1    &    1    & 1 \\ \hline
1      &  0    &   -1    & 2 \\
1      &  1    &    0    & 1 \\
1      & -2    &   -1    & 0 \\
\end{array}\right], \quad
\tF=\left[\begin{array}{c|ccc}
0      &  1    &    -2   & -1 \\ \hline
1      &  0    &   -1    & 1 \\
-2     &  1    &    0    & -1 \\
-1     & -1    &    1    & 0 \\
\end{array}\right].
\]
The eigenvalues are $\pm 1, \pm 2i$, so according to our convention stated in Remark
\ref{rem:choose} we choose $\sigma=1, \,\theta=-2$. In particular, the product $F\tF=-2I=(dh)I$.
Notice that $F$ needs the factor $1/\sqrt{5}$ to become normalized.\cvd
\end{example}

\begin{proposition}\label{note:F3}
$F^3=0$ for a hyper-singular $F$,  which  also entails the exponential \refrm{eZ}.{\rm (c)}.\cvd
\end{proposition}

As noted after \refrm{dual}, a Maxwell matrix $F$ generates a 4-dimensional algebra in which the
pattern of powers quickly becomes rather intricate. Therefore, a change to an adequate basis may
resolve this issue. For a normalized $F$ with eigenvalues $\pm\,\sigma,\,\pm\,i\theta$, and
choosing $(dh)=\sigma\theta$, we define
\[
Z=\theta\,F-\sigma\,\tF,
\]
which implies that $\tZ=\sigma \,F+\theta\,\tF$, together forming the rotation
\[
\bmat Z\\ \tZ \emat=
\left[\begin{array}{cc} \theta & -\sigma \\\sigma & \theta
\end{array}\right]
 \,\bmat F\\ \tF \emat.
\]
Applying the inverse rotation,
\[
F=\theta Z+\sigma\tZ \quad\mbox{and}\quad
\tF=-\sigma Z+\theta\tZ.
\]

Then, from the definition and  formulas \refrm{dual} we obtain the following relations which no
doubt simplify the patterns of \refrm{dual}, and can be extended easily to higher powers and,
consequently, to the exponential in an elementary way (cf.\ \refrm{eZ}).
\begin{equation}\label{ZZ}
\begin{array}{ll}
\mbox{(a)}\quad  &Z^3=-Z,\\
\mbox{(b)}\quad  &\tZ^3=\tZ,\\
\mbox{(c)}\quad &Z\tZ=\tZ Z=0.
\end{array}
\end{equation}
\begin{proposition}\label{note:unique}
In general, the representation $F=aX+bY$  is unique, where $a, b$ are nonzero, $XY=0$, $X^3=-X$ and
$Y^3=Y$.
 \end{proposition}

 {\em Proof}.
 The components satisfy also the linear equation $F^3=-a^3\,X+b^3\,Y$. Hence the underlying $2\times 2$ matrix is invertible, resulting in the unique solutions
\[
X=\frac{1}{a(a^2+b^2)}\left(b^2\,F-F^3\right)\quad \mbox{and}\quad Y=\frac{1}{b(a^2+b^2)}\left(a^2\,F+F^3\right).
\cvd\]
\begin{proposition}\label{prop:exp}
For a normalized $F$ we obtain the $G$-orthogonal exponentials
\[
e^{tF}=e^{t\theta Z}e^{t\sigma\tZ}=I+\sin t\theta\,Z+(1-\cos t\theta)\,Z^2+\sinh t\sigma\,\tZ+(\cosh t\sigma-1)\,\tZ^2.
\]
Alternatively, we may substitute
\begin{equation}\label{subst}
Z^2=F^2-\sigma^2\,I\quad\mbox{and}\quad \tZ^2=F^2+\theta^2\,I.
\end{equation}
\end{proposition}
\begin{proof}
The pattern of cubes \refrm{ZZ}, as noted in \refrm{eZ}, entails the formulas
\[
\begin{array}{rl}
e^{t \theta Z}&=I+\sin t\theta\,Z+(1-\cos t\theta\,)Z^2,\\
e^{t \sigma\tZ}&=I+\sinh t\sigma\,\tZ+(\cosh t\sigma -1)\tZ^2.
\end{array}
\]
The full formula follows since the orthogonal components commute. The substitution is a consequence
of \refrm{dual}.\end{proof}
\begin{example}\rm
The normalization, followed by the change of basis (the replacement of $F,\,\tF$ by $Z,\,\tZ$),
greatly simplify arguments and clarifies formulas. Yet, when it comes to numerical data one should
expect to pay a price in regard to appearance, which we will illustrate,  continuing  Example 2.3.
First we normalize $F$, dividing $\sigma=1$ and $\theta=-2$ by $\sqrt{5}=\sqrt{\sigma^2+\theta^2}$.
Hence the {normalized} $F$ and $\tilde{F}$ become
\[
F=\frac{1}{\sqrt{5}}\left[\begin{array}{c|ccc}
0      &  1    &    1    & 1 \\ \hline
1      &  0    &   -1    & 2 \\
1      &  1    &    0    & 1 \\
1      & -2    &   -1    & 0 \\
\end{array}\right], \quad
\tF=\frac{1}{\sqrt{5}}\left[\begin{array}{c|ccc}
0      &  1    &    -2   & -1 \\ \hline
1      &  0    &   -1    & 1 \\
-2     &  1    &    0    & -1 \\
-1     & -1    &    1    & 0 \\
\end{array}\right].
\]
The normalized $\sigma=\frac{1}{\sqrt{5}},\quad \theta=-\frac{2}{\sqrt{5}}$ entail the normalized
matrices $Z,\,\tZ$:
\[Z=\theta F-\sigma\tF =\frac{1}{5}\begin{bmatrix}
 0 & -3 & 0 & -1 \\
 -3 & 0 & 3 & -5 \\
 0 & -3 & 0 & -1 \\
 -1 & 5 & 1 & 0 \\
\end{bmatrix},\quad \tilde{Z}=\sigma F+\theta\tF=\frac{1}{5}\begin{bmatrix}
 0 & -1 & 5 & 3 \\
 -1 & 0 & 1 & 0 \\
 5 & -1 & 0 & 3 \\
 3 & 0 & -3 & 0 \\
\end{bmatrix}.
\]
The eigenvalues of $Z$ and $\tilde{Z}$ are, respectively,  $(i,-i,0,0)$ and $(-1,1,0,0).$ In order
to display exponentials explicitly we denote
 \[
 c=\cos \frac{2 t}{\sqrt{5}},\quad s=\sin \frac{2 t}{\sqrt{5}},\quad S= \sinh \frac{ t}{\sqrt{5}},\quad C= \cosh \frac{t}{\sqrt{5}}.
 \]
The direct computations yield the formulas
\begin{eqnarray*}\exp({t\theta Z})&=&I +\sin (t\theta ) Z+(1-\cos (t\theta)) Z^2\\&=&\frac{1}{5}\begin{bmatrix}
 -2 c+7 & c+3 s-1 & 2c-2 & -3 c+s+3 \\
 -c+3 s+1 & 5 c & c-3 s-1 & 5 s \\
-2c+2 & c+3 s-1 & 2 c+3 & -3 c+s+3 \\
 3 c+s-3 & -5 s & -3 c-s+3 & 5 c \\
\end{bmatrix},\end{eqnarray*}
\begin{eqnarray*}
\exp({t\sigma \tilde{Z}})&=&I +\sinh (t\sigma ) \tilde{Z}+\cosh (t\sigma)-I)\tilde{Z}^2\\&=&\frac{1}{5}\begin{bmatrix}
 7 {C}-2 & -C-S+1 & -2 {C}+5 {S}+2 & 3C+3S-3 \\
 C-S-1 & 5 & -C+S+1 & 0 \\
 2 {C}+5 {S}-2 & -C-S+1& 3 {C}+2 & 3C+3S-3 \\
-3C+3S+3& 0 & 3C-3S-3& 5 \\
\end{bmatrix}
\end{eqnarray*}
Without referring to the ``theory'' one may verify directly that product orthogonal $Z$ and $\tZ$
commute, so do these two one-parameter subgroups of the Lorentz group:
\[
\exp({t\theta Z})\exp({s\sigma \tilde{Z}})=\exp({s\sigma \tilde{Z}})\exp({t\theta Z}).
\cvd\]
\end{example}
\begin{corollary} \label{cor:props:EF}
Let $F$ be normalized. Then
\begin{itemize}
\item[{\rm (a)}] $\det(e^{tF})=1$;
\item[{\rm (b)}] denoting by $\delta_t$ the $(0,0)$-entry in the matrix $e^{tF}$,  we have
\[
\delta_t=1+(1-\cos t\theta)(d^2-\sigma^2)+(\cosh t\sigma-1)\,(d^2+\theta^2)\ge 1,
\]
and the equality occurs iff $\sigma=0$ ($\wek d\perp \wek h$) and $t\theta=2n\pi$, or $\wek
d=\wek 0$ and $\theta=0$;
\item[{\rm (c)}] $ \tr e^{tF}=2(\cos t\theta+\cosh t\sigma)\ge 0, $ with the equality occurring
    iff $t\theta=(2n+1)\pi$ and $\sigma=0$.
\end{itemize}
\end{corollary}
\begin{proof}
(a) is obvious, (b) follows from \refrm{subst} in virtue of {Proposition} \ref{eq:||}.(c), and
\refrm{subst} (or the shape of the eigenvalues $e^{\pm\,\sigma}, \,e^{\pm i\theta}$) imply
(c).\end{proof}
\subsection{More on orthogonal decomposition}\label{s:more ort}
Although the orthogonal decomposition is unique yet there are alternative ways to find the
components. Let $m$ denote the number of nonzero distinct eigenvalues of $F$ which we order into a
sequence  $\lambda_k, k=1,...,m$. For a diagonalizable matrix $F$, the diagonalization formula
$F=VDV^{-1}$ yields the orthogonal decomposition
\[
F=\sum_{k=1}^m \lambda_k X_k,
\]
where $X_k=VE_k V^{-1},\,k=1,...,m$, and $E_k$ are mutually orthogonal diagonal projections. For a
single nonzero eigenvalue $\lambda_k$, $E_k$ has the single $1$ at the $k$\tss{th} position. Since
the components $X_k$ commute and, like $E_k$, are mutually orthogonal projections, i.e.
$X_k^2=X_k$, then
\[
e^F=I+\sum_{k=1}^m \left(e^{\lambda_k}-1\right)\,X_k.
\]
We now rephrase Proposition \ref{note:unique}:
\begin{proposition}\label{le:orth}
A  diagonalizable  $G$-skew symmetric matrix $F=F(\wek d,\wek h)$ admits a unique orthogonal
decomposition $F=X+Y$ such that
\begin{equation}\label{eq:dec}
 X^3=\sigma^2\,X\quad\mbox{and}\quad Y^3=-\theta^2\,Y,
\end{equation}
where the $G$-skew symmetric components are given explicitly by the formulas
\begin{equation}\label{eq:XY}
X=\frac{\theta^2}{\sigma^2+\theta^2}\,F+\frac{1}{\sigma^2+\theta^2}\,F^3,\qquad
Y=\frac{\sigma^2}{\sigma^2+\theta^2}\,F-\frac{1}{\sigma^2+\theta^2}\,F^3.
\end{equation}
If $F$ is singular, then the decomposition is trivial, i.e., either $X=0$ or $Y=0$.
\end{proposition}
\begin{proof}
We pool the pairs of projections together
\[
X=\sigma(X_1-X_2),\quad Y=i\theta (Y_1-Y_2),
\]
and check directly their listed properties. Formulas \refrm{eq:XY} come as solutions of the
equations
\[
F=X+Y,\quad F^3=X^3+Y^3=\sigma^2\,X-\theta^2\,Y,
\]
and show that $X$ and $Y$ are $G$-skew symmetric.\end{proof}

\begin{corollary}\label{cor:R}
The exponential of a $G$-skew symmetric matrix is $G$-orthogonal and
\begin{equation}\label{eq:Rod}
\jor=e^{F}=I+\frac{\sinh\sigma}{\sigma}\,X+\frac{\cosh\sigma -1}{\sigma^2}\,X^2+\frac{\sin \theta}{\theta}\,Y+
\frac{1-\cos \theta}{\theta^2}\,Y^2.
\end{equation}
For a singular  $F$ with $\sigma^2+\theta^2>0$,  either the $X$-part or the $Y$-part vanishes. For
a hyper-singular $F$, we obtain
\begin{equation}\label{eq:exphyper}
e^F=I+F+\frac{1}{2}\,F^2.
\end{equation}
\end{corollary}
\begin{proof}
The regular pattern of the power series ensured by  \refrm{eq:dec} yields the exponential.
 By the same token, $G(\ct F)^nG=(-1)^n F^n$. Hence and by \refrm{eq:XY}, where the signs are changed only at the sine and  hyperbolic sine,  $G\ct{\jor}G=Ge^{\ct{F}}G=e^{-F}=\jor^{-1}$.

Due to continuity, the exponentials carry over to the simplified formulas as $\theta$ or $\sigma$
(or both) converge to 0. Also, we have already used the pattern $F^3=0$ in Proposition
\ref{note:F3}.
 \end{proof}

If $m$ is the number of nonzero distinct eigenvalues, then the components are solutions of the
linear  equations
\[
\sum_{k=1}^m\lambda_k^p \,X_k=F^p,\quad p=1,...,m,
\]
with the invertible  Vandermonde $m\times m$ matrix $M=[\lambda_k^p]$. In particular, the
exponential $e^F$ is a linear combination of linearly independent powers $F^p,\, p=0,\dots,m$.\vv

 Let $F$ be normalized and nonsingular, i.e., $m=4$. Then, displaying the Vandermonde matrix
\[
M=\bmat
\sigma   & -\sigma   & i\theta    & -i\theta  \\
\sigma^2 & \sigma^2  & -\theta^2  & -\theta^2 \\
\sigma^3 & -\sigma^2 & -i\theta^3 & i\theta^3 \\
\sigma^4 & \sigma^4  & -\theta^4  & \theta^4  \\
\emat,
\]
we arrive at its inverse, verifiable directly:
\[
M^{-1}=\frac{1}{2}\,
\bmat
\sigma^{-2} &    0        &    0        &   0        \\
    0       & \sigma^{-2} &    0        &   0        \\
    0       &    0        & \theta^{-2} &   0        \\
    0       &    0        &    0        &\theta^{-2}
\emat
\cdot
\left[\begin{array}{rrrr}
\sigma\theta^2   & \theta^2  & \sigma   & 1 \\
-\sigma\theta^2  &\theta^2   & -\sigma  & 1 \\
-i\sigma^2\theta & -\sigma^2 & i\theta  & 1 \\
i\sigma^2\theta  & -\sigma^2 & -i\theta & 1
\end{array}\right].
\]
In other words,
\begin{equation}\label{eq:Xs}
\begin{array}{rlcl}
\vspandexsmall X_1 &= \Frac{1}{2\sigma^2}\,
    \Big\{\,\left(F^4+\theta^2\, F^2\right) &+& \sigma\left(F^3+ \theta^2\,F\right)  \,\Big\},\\
\vspandexsmall X_2 &= \Frac{1}{2\sigma^2}\,
    \Big\{\,\left(F^4+\theta^2\,F^2\right) &-& \sigma\left(F^3+ \theta^2\,F\right)  \,\Big\},\\
\vspandexsmall X_3 &= \Frac{1}{2\theta^2}\,
    \Big\{\,\left(F^4 -\sigma^2\,F^2\right)&+& i\theta\left(F^3-\sigma^2\,F\right) \,\Big\},\\
          X_4 &= \Frac{1}{2\theta^2}\,
    \Big\{\,\left(F^4-\sigma^2\,F^2\right) &-& i\theta\left(F^3-\sigma^2\,F\right) \,\Big\}.
\end{array}
\end{equation}
As mentioned before, instead of the basis consisting of powers up to the fourth power, which may be
somewhat cumbersome to compute, we may switch to the simpler basis $I,\, F, \,F^2,\, \tF$.
\begin{corollary}
In the new basis, formulas \refrm{eq:Xs} read:
\[
\begin{array}{rlcl}
\vspandexsmall X_1 &= \Frac{1}{2}\,
    \Big\{\,\left(\theta^2\,I+ F^2\right) &+&\rule{3pt}{0pt} \left(\sigma\, F+ \theta\,\tF\right)  \,\Big\},\\
\vspandexsmall X_2 &= \Frac{1}{2}\,
    \Big\{\,\left(\theta^2\,I+ F^2\right) &-&\rule{3pt}{0pt} \left(\sigma\, F+ \theta\,\tF\right)  \,\Big\},\\
\vspandexsmall X_3 &= \Frac{1}{2}\,
    \Big\{\,\left(\sigma^2\,I-F^2\right)&-& i\left(\theta\,F-\sigma\,\tF\right) \,\Big\},\\
\vspandexsmall X_4 &= \Frac{1}{2}\,
    \Big\{\,\left(\sigma^2\,I-F^2\right)&+& i\left(\theta\,F-\sigma\,\tF\right) \,\Big\}.\\
\end{array}
\]
The exact formulas \refrm{eq:XY} now take an alternative form:
\[
X=\sigma^2\,Z+\sigma\theta\,\tZ,\qquad Y=\theta^2\,Z-\sigma\theta\,\tZ,
\]
which linearizes the exponential \refrm{eq:Rod}.
\end{corollary}
\begin{proof}
Use $F^3$ and compute $F^4$ from \refrm{dual}.\end{proof}
\subsection{Eigenvectors}\label{s:eigenvectors}
Recall that we choose $\sigma$ and $\theta$ to satisfy $(dh)=\sigma\theta$. In what follows
w.l.o.g. we may and do assume that $h=1$. To return to the general case (or to examine the case
$h=0$) it suffices to substitute $\sigma:= \sigma/h,\,\theta:=\theta/h$, then the case $h=0$ can be
handled in the limit. Relations \refrm{eqs} simplify further even already simple properties:
\begin{equation}\label{eq:proj}
\begin{array}{lll}
(\wek d\times\wek h)\times \wek h &=-h^2\wek d+(dh)\,\wek h  &=-\wek d+\sigma\theta\,\wek h,\\
\zeta\df ||\wek d\times\wek h||  &=\sqrt{d^2h^2-(dh)^2}   &=\sqrt{(1-\theta^2)(1+\sigma^2)}.
\end{array}
\end{equation}
\subsubsection{The regular case}\label{ss:regular}
First we assume that neither $\wek d$ and $\wek h$ are parallel,  nor they are orthogonal, i.e.,
$\zeta>0$ and $(dh)=\sigma\theta\neq 0$.  The assumption yields an orthonormal basis of $\R^4$:
\[
\wek {v}_0=               \bmat 1 \\ \bm 0\emat,\quad
\wek {v}_1=\bmat 0\\ \wek h\emat,\quad
\wek {v}_2=\frac{1}{\zeta}\bmat 0 \\ (\wek d\times\wek h)\times \wek h \emat,\quad
\wek {v}_3=\frac{1}{\zeta} \bmat0 \\ \wek d\times\wek h \emat.
\]
We immediately find the matrix representation of $F$:
\[
\left[\begin{array}{c|ccc}
0      &  (dh) &  -\zeta & 0 \\ \hline
(dh)   &  0    &    0    & 0 \\
-\zeta &  0    &    0    & 1 \\
0      &  0    &   -1    & 0 \\
\end{array}\right].
\]
To display eigenvectors we will use the function
\[
c=c(\sigma,\theta)=\sqrt{\frac{1-\theta^2}{1+\sigma^2}}.
\]
Recall that $(dh)=\sigma\theta$ and $h=1$. Then eigenvectors are:
\begin{equation}\label{eq:eigvZ:both:nice}
\begin{array}{rcl}
\vspandexsmall
\pm\,\sigma &\, \sim\,& \,\Big(\rule{5pt}{0pt}\wek {v}_0 +c\,\wek {v}_3\Big)\pm \rule{7pt}{0pt} \Big(\rule{6pt}{0pt}\theta\,\wek {v}_1-\,\sigma c\,\wek {v}_2\Big),\\
\pm\,i\theta &\, \sim\,& \,\Big(c\,\wek {v}_0+\rule{5pt}{0pt}\wek {v}_3 \Big)\mp\, i\,\,\Big(\sigma c\,\wek {v}_1+\rule{10pt}{0pt}\theta\,\wek {v}_2\Big).
\end{array}
\end{equation}
(Let us repeat that in order to pass from the normalized case $h=1$ to the general case it suffices
to substitute $\sigma:=\sigma/h,\theta:=\theta/h$.) By \refrm{eqs} and \refrm{eq:proj} the complex
Euclidean norm of both vectors equals $\sqrt{2}$.

\subsubsection{An insight into hyper-singularity}\label{ss:insight}
With exception of the quadruple null eigenvalue: $\sigma=\theta=0$, the boundary cases: $\wek
d\perp \wek h$ and $\wek d||\wek h$ are obtained in the limit from the above representations.  We
will specify these cases but first  let us comment on the issue.

The matrix $F=F_p$ is a continuous function of its vector parameter $p=(\wek d,\wek h)\in\R^6$.
Denote by  $E$ the set of all eigenvectors of $F_p$, for all $p$, and by $E^0$ the set of
eigenvectors belonging to the open set $P^0=\set{(dh)\neq 0\mbox{ and } \zeta>0}$. The question is
whether an eigenvector $\wek e_p$ is a cluster point of $E^0$. If this happens, we may call the
system {\bf connected} at $p$. In the opposite case, not only a matrix is singular but the
eigenvector is separated from the rest of potential eigenvectors, justifying again the name
``hyper-singular''.

\begin{enumerate}
 \item  {\bf The connected case}
 \begin{enumerate}
\item $d=0$, so $\theta=h, \sigma=0$. Let $\wek h_1,\wek h,\wek h_3$ form a positively oriented
    orthonormal basis of $\R^3$. Then
\begin{equation}\label{d=0}
0\quad\sim\quad\bmat 0 \\ \wek h \emat \quad\mbox{and}\quad \bmat 1 \\ \wek 0 \emat,
\qquad
i\theta\quad\sim\quad \bmat 0 \\ \wek h_1 \emat+i\bmat 0 \\ \wek h_3 \emat.
 \end{equation}
  Now let $h=1$ and let $\sigma\to 0$ in \refrm{eq:eigvZ:both:nice}. Hence $\theta\to 1$, so
  $c\to 0$ and $\zeta\to 0$. The eigenvectors for $\pm\,\sigma$ actually converge to the
  vectors listed above. However, the normalized vectors $\wek {v}_2$ and $\wek {v}_3$ move
  within the unit circle on the plane orthogonal to $\wek h$, mapping a path along which $\wek
  d$ converges to $\wek 0$.

  In other words, the given particular vectors may not converge but their eigenspaces do. In
  fact, by a routine compactness argument there exists a discrete orbit convergent to the
  complex vector such as one listed above.

 \item $h=0$, so $\theta=0,d=\sigma$. Choose two orthonormal vectors $\wek d^\perp$  orthogonal
     do $\wek d$. Then
\[
0\quad\sim\quad\bmat0 \\ \wek d^\perp \emat,
\qquad
\pm\,\sigma \quad\sim\quad \bmat \pm\,d \\ \wek d \emat.
 \]

By duality $Z\mapsto \tilde{Z}$ and the first part the system is connected here.
\item $d>0,\,h>0$, $(dh)=0$ but only one eigenvalue is 0. Formula \refrm{eq:eigvZ:both:nice}
    covers this case. Proposition that the undetermined ratio $\frac{\sigma\theta}{\lambda}$
    could be interpreted by continuity if we adopt the convention $\frac{0}{0}=1$.
\item $\wek d||\wek h$ and  $d>0,\,h>0$.  Let us augment $\wek h$ to  form a positively
    oriented orthonormal basis $\wek h,\wek h_2,\wek h_3$ in $\R^3$. Then
\[
\sigma=\pm\,d\quad\sim\quad\bmat \sigma \\ \wek d \emat,
\qquad
i\theta=ih\quad\sim\quad \bmat 0\\ \wek h_2 \emat +i \bmat 0 \\ \wek h_3 \emat.
\]
By \refrm{eq:eigvZ:both:nice}, since $c^2=\zeta^2/(1+\sigma^2)\to 0$ as $\zeta\to 0$, the
system is connnected here.
\end{enumerate}
 \item {\bf The disconnected case}

 $d>0,\,h> 0$ but $\sigma=\theta=0$. Then
\[
0\quad\sim\quad\bmat 0 \\ \wek h \emat
\quad\mbox{and}\quad
\bmat h^2 \\ \wek d\cross\wek h \emat.
\]

Let $h=1$ and $\theta,\sigma\to 0$. Then
$d^2-\sigma^2\theta^2=\zeta^2-(1-\theta^2)(1+\sigma^2)\to 1$. Hence $d\to 1$. Also, $c\to 1$.
Thus all eigenvectors in \refrm{eq:eigvZ:both:nice} converge to $\wek {v}_3+\wek {v}_0$,  the
second vector listed above, up to a multiplier.

Only the second eigenvector is preserved in the limit and the first eigenvector is separated,
i.e., the complex Euclidean distance between $\wek {v}_1$ and the span of either eigenvector is
1.
\end{enumerate}
\section{Eigen-system of a Lorentz matrix $\jor$}\label{sect:Lorentz}
We will examine a Lorentz matrix by studying its components, focusing on the eigenvalues and
eigenvectors.   We observed previously (cf.  Subsection \ref{ss:insight}) that a component of a
continuous ``orderly'' matrix may generate singularities or discontinuities. We will encounter this
phenomenon again. For example, the eigen-system of the matrix $A$, described below, will also be
disconnected on the boundary. A hidden rotation $R$ in the interior of the manifold allows only
angles $\rho\in (-\pi/2, \pi/2)$, yet on the boundary, where $s=1$, it will exhibit a jump to the
rotation by $\rho=\pi$.

In this section we do not restrict the dimension $d$ of the Euclidean space until it will be
required by an argument. The results are still elementary, partially due to our avoidance of the
sophisticated context of Lie algebras, Lie groups, and the general matrix theory (cf. also
\cite{bartocci2014}).

\subsection{$G$-orthogonal matrix}\label{s:Gort}
Again, a  matrix $\jor$ is $G$-orthogonal if $\cj{\jor}=G\ct {\jor}G=\jor^{-1}$. That is,
\begin{equation}\label{eq:jort:def}
\mbox{(a)}\quad G\jor\ct{G}\jor=I\quad\mbox{and/or}\quad \mbox{(b)}\quad \jor G \ct{\jor} G=I.
\end{equation}
Clearly, $\det \jor\neq 0$ and, since an eigenvalue $\lambda$ entails the eigenvalue $1/\lambda$,
$|\det(\jor)|=1$.

The product of $G$-orthogonal matrices is $G$-orthogonal. Let us consider the 1:3 (time-space)
split of a $G$-orthogonal matrix:
\[
\jor=\bmat {s} & \ct{\wek q}\\ \wek p & {A} \emat.
\]
Definition \refrm{eq:jort:def} gives characterizing equations
\begin{equation}\label{eq:jort:eqs}
\begin{array}{lrcllcl}
\mbox{(i)}  & \ct {A}\wek p &=& {s}\,\wek q, & {A}\wek q &=& {s}\, \wek p, \\
\mbox{(ii)} & \ct {A}{A}     &=& \wek q\ct {\wek q}+I, & {A}\ct {A}    &=& \wek p\ct {\wek p}+I,\\
\mbox{(iii)}& p^2     &=& {s}^2-1,     & q^2    &=& {s}^2-1,   \\
\end{array}
\end{equation}
where both columns are deducible from each other. Obviously, ${s}^2\ge 1$ and
\[
{\rm \det}\!^2({A}) =2{s}^2-1\quad\mbox{and}\quad {A}^{-1}=\ct {A}-\frac{1}{{s}}\, \wek q\ct{\wek p}.
\]
By \refrm{eq:jort:eqs}.(ii), the eigen-system of $\ct {A}{A}$ and ${A}\ct {A}$ is
\[
\begin{array}{rll}
\ct {A}{A}: & \quad 1\quad\sim\quad \wek q^\perp,\qquad & {s}^2\quad\sim\quad \wek q,\\
{A}\ct {A}: & \quad 1\quad\sim\quad \wek p^\perp,\qquad & {s}^2\quad\sim\quad \wek p.\\
\end{array}
\]
\begin{proposition}
${A}$ is orthogonal iff $|{s}|=1$ iff $\,\wek p=\wek q=\wek 0$.
\end{proposition}
\begin{proof}
Clearly, if $|{s}|=1$ or either vector is zero (hence both), then $\ct {A}{A}={A}\ct {A}=I$.
Conversely, if ${A}$ is orthogonal, then $b^2=\wek p\ct {\wek p}=0$, i.e., $\wek p=\wek q=\wek 0$,
and $|{s}|=|{s}\,\det({A})|=|\det(\jor)|=1$.
\end{proof}
\begin{proposition}
Further, either of the above conditions, augmented by the requirement $\det(\jor)=1$ and ${s}>0$,
occurs iff  ${A}=e^{\theta V(\mbox{\footnotesize $\wek h$})}$ and ${s}=1$. In particular, there is
a $G$-skew symmetric matrix $F$, not unique, such that $\jor =e^F$.
\end{proposition}
\begin{proof}
The additional condition just states the reduction to the classical 3D Rodrigues exponential
formula, cf. Example \ref{ex:Rod}. In particular, it ensured $\det (A)=1$. Thus, there are
infinitely many choices for $F$ as the exponent, cf. Cor. \ref{expsurj}. \end{proof}

\begin{proposition}\label{prop: A'A}
If ${s}>0$ then $\ct{\jor}\jor$ is diagonalizable:  $\ct{\jor}\jor=V \Delta V^{-1}$ with
$\Delta=\diag(\gamma,1/\gamma,1,1)$ for some $\gamma>0$, subject to relations
\begin{equation}\label{eq:AA}
\alpha^2={s}^2-1,\qquad
\alpha=\frac{\gamma-1}{\gamma+1}, \qquad \gamma=\frac{1+\alpha}{1-\alpha},
\end{equation}
with ${s}=1$ iff  $\alpha=0$ iff  $\gamma=1$. Four independent eigenvectors, of which two are owned
by the eigenvalue 1, are

\begin{equation}\label{eq:eig AA}
 1\quad\sim\quad \bmat 0 \\ \wek q^\perp \emat,\qquad
 \gamma\quad\sim\quad \bmat \alpha \\ \wek q \emat,\qquad\quad
\quad \frac{1}{\gamma}\quad\sim\quad \bmat -\alpha \\  \wek q \emat.
\end{equation}
\end{proposition}
\begin{proof}
We check directly that
\[
\ct{\jor}\jor=2\bmat {s} \\ \wek q\\ \emat \,\bmat {s} & \ct {\wek q} \emat+G,\label{eq:rtr}
\]
which yields two eigenvectors with $\wek q^\perp$, orthogonal to $\wek q$, owned by the unit.  The
eigen-equation
\[
\ct{\jor}\jor\bmat \alpha \\ \wek q \emat=\gamma \bmat \alpha \\ \wek q \emat
 \]
entails a system of simple equations
\[
\begin{array}{rl}
2\left({s}^2-1+{s}\alpha\right)+1&=\gamma,\\
2\left({s}^2-1+{s}\alpha\right)\,{s}-\alpha&=\alpha\gamma,\\
\end{array}
\]
which yields immediately \refrm{eq:AA} and the remainder of \refrm{eq:eig AA}.
\end{proof}

Henceforth, we are confining ourselves to proper Lorentz matrices, i.e., to $G$-orthogonal matrices
such that $\det (\jor)=1$ and ${s}>0$.

\subsection{Sliders}\label{s:sliders}
A unit vector $\wek u$ in a  Euclidean space and $t\in\R$ induce a group of commuting operators
\[
S_t=S(\wek u;t)=I+(t-1)\,\wek u\wek u',
\]
called {\bf sliders} for the following reason:
\[
S_t\wek x=\wek x+(t-1)\,(xu)\,\wek u=\wek x+(t-1)\,{\sf proj}_{\wek u}\wek x.
\]
The product rule $ S_t\, S_s=S_{ts}$ entails the powers $S_t^n=S_{t^n}$. Hence, the inverse for
$t\neq 0$ is  $S_t^{-1}=S_{1/t}$ and sliders $S_t$ are ``{\em infinitely divisible}'' for $t>0$:
\[
S_t^{1/n}\df S_{t^{1/n}}\quad \mbox{because}\quad  S_{(t^{1/n})^n}=S_t.
\]
We choose the positive sign when $n$ is even although, as expected,
\[
J^2_{\sqrt{t}}=S_t\quad\mbox{and}\quad J_{-\sqrt{t}}^2=S_t,\quad t>0.
\]
Its eigensystem emerges immediately; the eigenvalue $t\,\sim\, \wek u$ and the eigenvalue
$1\,\sim\, \wek u^\perp$. Hence, in a $d$-dimensional space,
\[
\det(S_t)=t,\qquad \tr S_t=d-1+t,\qquad ||S_t||^2=1+\frac{t^2-1}{d}.
\]
The product rule upon the substitution $s=e^\sigma,\,t=e^\tau$ shows that sliders are exponentials
of projections $\wek u\ct{\wek u}$:
\[
e^{\sigma\,\wek u\ct{\wek u}}=I+\left(e^\sigma-1\right)\,\wek u\ct{\wek u}=S_{e^\sigma}\df T_\sigma.
\]
That is,  $T_{\sigma}T_{\tau}=T_{\sigma+\tau}$ in the additive mode.

\subsection{Polar representation}\label{s:polar}
Let us rewrite a $G$-orthogonal (or Lorentz) matrix in a normalized ``$1\!+\!(d\!-\!1)$''  form
\[
\jor=\bmat s &{t}\, \ct{\wek v} \\{t}\, \wek u & A\emat,\quad ||\wek u||=||\wek v||=1,\quad {t}^2={s}^2-1,
\]
characterized by properties \refrm{eq:jort:eqs}, which now take the following appearance:
\be\label{char Lambda} |{s}|\ge 1,\quad \ct AA=S_{{s}^2}(\wek v),\quad A\ct A=S_{{s}^2}(\wek
u),\quad A\wek v={s}\,\wek u,\quad \ct A\wek u={s}\,\wek v. \ee We may choose the sign of ${t}$ at
will, changing directions of both vectors $\wek u$ and $\wek v$, or alternatively,  by switching to
$G\jor G$. Our main interest lies in {\bf proper Lorentz matrices}, i.e., with $\det\jor=1$ and
${s}>0$. We will introduce these assumptions gradually.

We infer quickly that $|\det A|={s}$ and find the inverse from either formula:
\[
A^{-1}=S_{{s}^{-2}}(\wek v)\,\ct A=\ct A S_{{s}^{-2}}(\wek u).
\]

The unique polar decomposition $A=RS$ follows immediately: \be\label{A=RS} S\df\sqrt{\ct A
A}=S_{{s}}(\wek v)=I+({s}-1)\,\wek v\ct{\wek v},\qquad R\df AS^{-1}=A-({s}-1)\,\wek u\ct{\wek v},
\ee where $R$ is orthogonal and $R\wek v=\wek u$.  We see that $\sign{\det R}=\sign{\det A}$.

Whence the unique polar decomposition emerges: \be\label{UP} \jor=UP,\quad\mbox{with}\quad
 U=\bmat 1 & 0\\ 0 & R\emat
\quad\AND\quad P=\bmat {s} &{t}\,\ct{\wek v}\\{t}\, \wek v &S\emat. \ee For a proper Lorentz
matrix, when ${s}\ge 1$ and $\det \jor=1$, we see that  $1=\sign{\det \jor}=\sign {\det U} =\sign
{\det R}=\sign {\det A}$. Thus $\det A={s}$ and  $R$ is a rotation by an angle $\rho$ about a
direction $\wek r$.

\begin{proposition}\label{Peigen}
The eigenvalues of $P=\sqrt{\ct{\jor}\jor}$ consist of the $(d-2)$-tuple 1 and the pair of positive
mutual reciprocals $\gamma, \,1/\gamma$, with the eigen-system
\[
 1\,\sim\, \bmat 0\\\wek v^\perp\emat,\qquad
 \gamma={s}+{t}\,\sim\, \bmat 1\\\wek v\emat,\qquad
\frac{1}{\gamma}={s}-{t}\,\sim\, \bmat \wek v\\ -1\emat.
\]
Moreover, $\gamma=1$ iff ${s}=1$.
\end{proposition}
\begin{proof}
 An eigenvector $\bmat s\\\wek v\emat$ with $\xi=0$ yields $\wek v^{\perp}$, and for $\xi\neq 0$ the eigen-equation
\[
\gamma ={s}+ \xi{t},\quad \xi\gamma ={t}+{s} \xi,
\]
is solved by $\xi^2=1$ and $\gamma_s={s}+\xi{t}$, as displayed.\end{proof}

\subsection{The intrinsic pattern}\label{s:intrinsic}
It is tempting to reverse the process  \refrm{A=RS}, given the parameters ${s}$, $R$, $\wek v$:
\[
\wek u\df R\wek v, \qquad A\df R+({s}-1)\,\wek u\ct{\wek v}, \qquad {t}=\pm\sqrt{{s}^2-1}.
\]
Then one might expect to recover the $G$-orthogonal matrix $\jor$ in a quite trivial way. However,
the postulates $\det \jor=1$ and ${s}\ge 1$ impose intrinsic relations  between parameters.
Therefore, their pattern should be described first. Let us begin with the simplest cases.
\begin{proposition}\label{s=1}{~}
\begin{itemize}
\item[{\rm (a)}] ${A}$ is orthogonal iff $|{s}|=1$ iff ${t}=0$.
\item[{\rm (b)}] $A$ is a rotation about a direction $\wek r$ iff ${s}=1$ and $\det \jor=1$.
    Equivalently, $A=e^{-\theta\,C(\wek r)}$ and ${s}=1$.
\item[{\rm (c)}] If $s=1$ then $\jor$ is diagonalizable, with the double eigenvalue 1 that owns
    $\bmat 0\\\wek r\emat$ and $\bmat 1\\\wek 0\emat$.
\end{itemize}
\end{proposition}
\begin{proof}
(a) follows from relations $\ct AA=I+({s}^2-1)\,\wek v\ct{\wek v},\, A\ct A=I+({s}^2-1)\,\wek
u\ct{\wek u}$ in \refrm{char Lambda}.

The first part of (b) follows from the remark after \refrm{UP} regarding the signs of determinants.
The next parts simply refer to the classical Rodrigues formula, cf. Example \ref{ex:Rod}. While a
rotation is involved, the uniqueness occurs only up to the periodic translation $\theta\mapsto
\theta +2n\pi$. The last part states the obvious.
 \end{proof}

Henceforth we assume that ${s}>1$ but we will be monitoring the limit behavior for ${s}\searrow 1$
as means of control. We also confine to four dimensions, $d=4$. Consequently, $A=RS$ has at least
one real eigenvalue $\alpha$ that owns a real unit eigenvector $\wek a$. Let $\rho$ denote the
angle of rotation $R$ about a direction $\wek r$, i.e., $\cos\rho=\ct {\wek x}\ct R \wek x$ for any
vector $\wek x\neq \wek r$.

Let us rewrite the eigen-equation $RS\,\wek a=A\,\wek a=\alpha\,\wek a$: \be\label{Rx} \wek
a+({s}-1)\,(av)\,\wek v=\alpha\,\ct R\wek a. \ee

\begin{proposition}\label{aquad}
Let $\jor$ be a proper Lorentz matrix.
\begin{itemize}
\item[{\rm (a)}] Then $(av)^2=\Frac{\alpha^2-1}{{s}^2-1}$, which implies that $1\le |\alpha|\le
    {s}$.
\item[{\rm (b)}] If $|\alpha|=1$ then $(av)=0$. In this case eigenvalues of $A$ belong to the set
    $\set{-1,1}$, and:
 \begin{itemize}
 \item[{\rm (a)}] either we have $(1,1,1)$, i.e., $A=I$ and $\jor=I$,
  \item[{\rm (b)}] or  we have $(-1,-1,1)$, i.e., $s=1$ and $A$ is a rotation by $\pi$ about
      $\wek a$.
  \end{itemize}
\item[{\rm (c)}] If $|\alpha|\neq 1$ then $\alpha$ satisfies the quadratic equation
    $\alpha^2-\alpha\,({s}+1)\,\cos\rho+{s}=0$, i.e., \be\label{alpha}
    \alpha=\frac{({s}+1)\,\cos\rho\pm\,\sqrt{({s}+1)^2\,\cos^2\rho-4{s}}} {2}. \ee
\end{itemize}
\end{proposition}
\begin{proof} Comparing the squared lengths of the vectors on both sides of \refrm{Rx},
\[
1+2 ({s}-1)\,(av)^2+({s}-1)^2\,(av)^2=\alpha^2,
\]
we deduce (a), and then (b). Then (c) follows by multiplying \refrm{Rx} by $\ct{\wek a}$, which
yields the formula
\[
1+({s}-1)\,(vx)^2=\alpha\,\cos\rho,
\]
and then the listed equation by (a).
\end{proof}

\begin{corollary}\label{cosrho}
Let $\jor$ be proper with ${s}>1$.
\begin{itemize}
\item[{\rm (a)}] Let $|\cos\rho|=1$.  Then
\begin{itemize}
\item[{\rm (i)}] either $\rho=0$, i.e., $R=I$, and $1,\,1,\,{s}$ are the eigenvalues of $A$.
    That is, $A$ reduces to the slider, $A=S_{{s}}(\wek v)$;
\item[{\rm (ii)}] or $\rho=\pi$, i.e., $A$ is a rotation by $\pi$ as in Proposition
    \ref{aquad}.2.(b).
\end{itemize}
\item[{\rm (b)}] Let $|\cos\rho|<1$. Then  $A$ has only one positive real eigenvalue
    $\alpha=\sqrt{{s}}$ and
\be\label{restr} \cos\rho= \frac{2\sqrt{{s}}}{{s}+1}>0. \ee
\end{itemize}
\end{corollary}
\begin{proof}
The first statement is evident, so assume that $|\cos\rho|<1$. Since a real eigenvalue exists, the
radicand in \refrm{alpha} must be nonnegative. If it is zero, then the statement follows. To
complete the proof we must exclude the remaining case. Assume by contrary  that the radicand is
strictly positive. We will see that this assumption implies that $\cos\rho=1$, which has been
excluded.

If $A$ had two more real eigenvalues, then one of them would have to be double.  But their product
equals ${s}=\det A$, hence  one of the roots would be $1$.  This occurs iff $\cos\rho=1$.

Alternatively, a single positive real eigenvalue of $A$ would be one of the roots while the second
root would be a phantom.  Imposing the restriction $1\le \alpha\le {s}$ on either root would entail
again $\cos\rho=1$.
\end{proof}

\begin{corollary}\label{Anormal}
$A$ is normal iff $\wek u||\wek v$ iff $\wek r||\wek v$. In either case, if $\jor$ is proper, then
${s}=1$ and $A$ reduces to a rotation.
\end{corollary}
\begin{proof}
The first statement follows directly from \refrm{char Lambda}. Then ${s}$ becomes an eigenvalue. It
is positive for a proper $\jor$, so ${s}=\sqrt{{s}}$, i.e., ${s}=1$. Proposition \ref{s=1} finishes
the argument. \end{proof}

\subsection{Eigenvalues of $\jor$}\label{s:eigenvalues of Lorentz}
We have already derived the polar representation \refrm{UP} with ${t}^2={s}^2-1$. Since the case
$s=1$ reduces to the classical Rodrigues formula, for the remainder of the subsection we assume
that $s>1$.

Based on \refrm{UP}, we write the eigenvalue equation
\[
\lambda \,\bmat\varepsilon\\ \wek x \emat =\jor \bmat \varepsilon\\ \wek x\emat=
\bmat 1 & 0\\ 0 & R\emat \,\bmat {t}\,(vx)+\varepsilon{s}\\ S \wek x+\varepsilon\,{t}\,\wek v \emat=
\bmat{t}(vx)+\varepsilon\,{s}\\ RS \wek x+\varepsilon\,{t}\,R\wek v\emat,
\]
requesting that $x=1$, i.e., $\wek x$ is a unit vector. As before in the case of $A$, we rewrite
the equation in a more transparent form, using  \refrm{A=RS}, $R^{-1}=\ct R$:
 \[
\wek x +\Big((s-1)(vx)+\varepsilon t\Big)\, \wek v=\lambda\,\ct R\,\wek x,\qquad {t}\,(vx)=\varepsilon(\lambda-{s}),
\]
or equivalently,
\[
\wek x+\frac{\varepsilon}{t} \,\Big({({s}-1)(\lambda-{s})+t^2}\Big)\,\wek v=\lambda\, \ct R\wek x,\qquad {t}\,(vx)=\varepsilon(\lambda-{s}),
\]
or even simpler: \be\label{l2} \wek x+\varepsilon \,\Frac{({s}-1)(\lambda+1)}{{t}}\,\wek
v=\lambda\, \ct R\wek x,\qquad {t}\,(vx)=\varepsilon(\lambda-{s}). \ee

\begin{proposition}\label{eps=1} Let ${s}>1$.
\begin{itemize}
\item[{\rm (a)}] If $\lambda\neq -1$, then $|\varepsilon|=1$.
\item[{\rm (b)}] Let  $\lambda=-1$. Then $R$ is the rotation by $\rho=\pi$ about some direction
    $\wek r$, and
\begin{itemize}
\item[{\rm (i)}] either $\wek r=\wek v$ and $\varepsilon=0$; so $\lambda=-1$ owns $\bmat 0\\
    \wek v^\perp\emat$;
\item[{\rm (ii)}] or $\wek v\perp \wek r$ and $\lambda=-1$ owns
 \[
\bmat \varepsilon\\ \wek v\emat, \mbox{ with } \varepsilon=-\sqrt{\Frac{s-1}{s+1}}, \qquad\mbox{and also owns}\quad
 \bmat 0\\ \wek r\cross\wek v\emat.
 \]
\end{itemize}
\end{itemize}
\end{proposition}
\begin{proof}
(a): When $\lambda=s$, then $\wek v\perp\wek x$ and the resulting right triangle yields
\[
\cos\rho=\frac{1}{s},\qquad \varepsilon^2=1.
\]
Squaring the lengths in \refrm{l2} and the formula for $(vx)$ give
\[
1+\varepsilon^2\left( \frac{2(\lambda+1)(\lambda-{s})}{{s}+1}+
\frac{(\lambda+1)^2({s}^2-1)}{({s}+1)^2}\right)=\lambda^2,
\]
and again we end up with $\varepsilon^2=1$.

(b): For $\lambda=-1$ formula \refrm{l2} reads
\[
R\wek x=-\wek x,\qquad (vx)=-\varepsilon\,\sqrt{\frac{s+1}{s-1}},
\]
so $R$ is the rotation by $\pi$ about $\wek r$. $\varepsilon=0$ yields all vectors orthogonal to
$\wek v$, so $\wek v=\wek r$. If $\varepsilon\neq 0$, then we would have infinitely many
eigenvector unless $\wek x=\wek v\perp\wek r$, so the formula for $\varepsilon$ follows. The
remaining eigenvector is made by $\wek r\cross\wek v$ and $\varepsilon=0$.
\end{proof}

\begin{corollary}
Let $\jor$ be a Lorentz matrix (not necessarily proper). Let $\rho$ denote the angle between $\wek
x$ and $R\,\wek x$ ($R$ might not be a rotation).

Then there exist two positive eigenvalues $\lambda=a\pm\,\sqrt{a^2-1}$, reciprocal to each  other,
that  solve the quadratic equation
\[
\lambda^2-2a\,\lambda+1=0,
\qquad\mbox{where}\quad a=\frac{({s}+1)\cos\rho+{s}-1}{2}.
\]
Of course, this occurs if and only if $a\ge 1$.
\end{corollary}
\begin{proof}
Multiply \refrm{l2} by $\ct{\wek x}$ and substitute $(vx)$, using $\varepsilon^2=1$:
 \[
 1+\Frac{(\lambda+1)(\lambda-{s})}{{s}+1}=\lambda\,\cos\rho.
 \]
In other words,
\[
(\lambda\cos\rho-1)({s}+1)= (\lambda-{s})\,(\lambda+1),
\]
which translates to the above quadratic equation. For two positive solutions to exist, $a$ must be
positive. Thus, it is necessary and sufficient that $a\ge 1$, which is rewritten above.
\end{proof}
\begin{corollary}
Let $\jor$ be proper with ${s}>1$. Then $\jor$ has two distinct positive eigenvalues, mutually
reciprocal.
\end{corollary}
\begin{proof}
Corollary \ref{cosrho} concludes with the intrinsic relation \refrm{restr}. Also, the assumption
ensures  ${s}>0$ and $\cos\rho>0$. That then condition $a\ge1$ translates to
\[
\cos\rho\ge\frac{3-{s}}{{s}+1},
\]
which is satisfied because by \refrm{restr}
\[
\cos\rho=\frac{2\sqrt{{s}}}{{s}+1}\ge \frac{3-{s}}{{s}+1}
\]
and the latter equality occurs iff ${s}=1$. However, the double root would require $\cos
\rho=(3-s)/(s+1)$, and it is excluded in virtue of the assumption $s>1$.
\end{proof}

\begin{corollary}
With the exception of the single eigenvalue 1, a proper $\jor$ is diagonalizable.\cvd
\end{corollary}
\begin{corollary}\label{expsurj}
Every proper Lorentz matrix is an exponential $\jor=e^F$ of a $G$-skew orthogonal $F$. The
parameter $\theta$ is unique up to a periodic shift: once a $\theta$ is given, then all
$\theta+2n\pi$ will do.
\end{corollary}
\begin{proof}
First, consider a defective $\jor$. Suppose that $P=\sqrt{\ct{\jor}\jor}$  has a quadruple
eigenvalue 1, i.e., $\gamma=1$. Since $2({s}+1)=\tr P=4$, hence ${s}=1$ and $\wek u=\wek v=0$. So,
$P=I$ and we have the classical Rodrigues  formula with a rotation $R$ by an angle $\rho$ about a
unit vector $\wek r$:
\[
\jor=U=\bmat 1 & \ct{\wek 0}\\ \wek 0 &R\emat,\quad \mbox{i.e.}\quad R=e^{-\rho C(\wek r)}.
\]
Besides this case, $\jor$ is diagonalizable with eigenvalues $s,\overline{s},
e^\sigma,e^{-\sigma}$, where $\sigma=\frac{1}{2}\ln\gamma$. Thus $s=e^{i\theta}$. The real $s=\pm
1$ correspond to $\theta=0$ or $\theta=\pi$, owning two linearly independent eigenvectors.  A
diagonalization $\jor=V e^D V^{-1}$, where $D=\diag(i\theta,-i\theta,\sigma,-\sigma)$ entails the
family of matrices
\[
\jor(t)=V e^{tD} V^{-1},\quad t\in\R.
\]
Since $V$ does not depend on $t$, we obtain a $G$-skew symmetric
\[
F\df \lim_{t\to 0}\, \frac{1}{t}\,\Big(\Lambda(t)-I\Big)=V DV^{-1}.
\]
Clearly, $e^F=\jor$.\end{proof}
\section{\bf Complex coding}\label{sect:Pauli coding}
We use the script font to denote classes (families, spaces, ideals, etc.) of linear operators on a
complex separable Hilbert space $\BH$  with a standard orthonormal basis $\wek e_0,\,\wek
e_1,\,\wek e_2,...$. Unless specifically stated, we consider only finite dimensional $\BH=\BC^d$.
 We may indicate the dimension of the underlying Euclidean space if necessary or skip it when it is clear from the context.
$\cC=\cC^d=\cL(\BC^d)$ denotes the space of all complex $d\times d$ matrices while
$\cH=\cH^d=\cL_H(\BC^d)$ marks its real subspace of Hermitian matrices. We can write
$\cC=\cH+i\cH$, i.e., every complex matrix $C$ can be written as a complex combination of Hermitian
matrices: \be\label{RR} C=H_1+i\,H_2, \ee where
\[
H_1=\frac{1}{2}\Big(C+C^*\Big)\AND H_2=\frac{1}{2i}\Big(C-C^*\Big).
\]
If the Euclidean space is $d$-dimensional, then $n\df {\rm dim}\, \cH=d^2$. In particular, if we
expand both matrices with respect to a basis $(\sigma_k)$ in $\cH$, with real coefficients
$(u^j_k), \,j=1,2$, then $X$ can be coded uniquely as the pair $(\wek u_1,\wek u_2)$ of two real
vectors, cf. also \refrm{uiv} below.

\subsection{Operators on operators and trace}\label{s:tens}
Factually, will be considering the tensor products of operators. However, in order to preserve the
elementary level of the presentation we avoid the symbolics and farther leads to the depths of the
theory. We simply answer the question ``how'' rather than ``why''.

 Since the matrix $\wek x\wek y^*$ can be viewed as an operator of rank one acting on vectors, the  operators $E_{jk}=\wek e_j\wek e_k^*$  form a basis of the vector space of operators  $\cL(\BH)$.  The basis is ordered lexicographically: $E_{11},\dots,E_{1,n},\, E_{21},\dots, E_{2,n},\,\dots_, E_{n1},\dots, E_{nn}$. We define $\tr E_{jk}\df\delta_j^k$, extendable by linearity. That is, for ${C}=\sum_{j,k} c_{jk} E_{jk}$, we put $\tr {C} =\sum_j c_{jj}$. Of course, one must show that the trace is well-defined, i.e., it is independent of the basis.

By $E_{jk,pq}$ we denote the operator acting on matrices, i.e., on operators, by its action on the
spanning set $\set{\wek x\wek y^*}$ of $\cL(\BH)$,
\[
E_{jk,pq} \, \wek x\wek y^*
=x_k\,y_q\,\wek e_j\wek e_p^*,
\]
extendable by linearity. In other words,
\[
E_{jk,pq}(X)=E_{jk} X E_{pq}.
\]
These operators form a basis of $\cL(\cL(\BH))$ on which the trace is defined by the formula
\[
\tr E_{jk,pq}\df \delta_j^k\cdot\delta_p^q.
\]
That is, since an operator $L\in \cL(\cL(\BH))$ is a linear combination of such basic operators
with coefficients, say, $l_{jk,pq}$, then \be\label{jjpp} \tr L =\sum_{j,p} l_{jj,pp}. \ee The
index-separated coefficients $l_{jk,pq}=c_{jk}\overline{d_{pq}}$ yield the ``jaws'' operators
\[
L_{C,D} X= {C} X {D}^*,\quad L_C\df L_{C,C}.
\]
\begin{proposition} \label{note:props L{C}}
Immediate properties:
\begin{itemize}
\item[{\rm(a)}] $L_{C}$  preserves the Hermicity.
\item[{\rm(b)}] The ``jaws composition'' holds:  $L_{C} L_D =L_{{C}D}$.
\item[{\rm(c)}] $\set{L_{C}: \det {C}\neq 0}$ is a group with $L^{-1}_{C}=L_{{C}^{-1}}$.

Further, $L_{C}$ and ${C}$ are simultaneously nonsingular or singular.
\item[{\rm(d)}] $L^*_{C}=L_{{C}^*}$;
\item[{\rm(e)}] $L_{C}$ has a nonnegative trace and further \be\label{trL} \tr L_{C} =\sum_{j,p}
    c_{jj}\overline{c_{pp}}= |\tr {C}|^2, \quad \tr L_{C,D} =\sum_{j,p} c_{jj}\overline{d_{pp}}=
    \tr {C} \cdot\overline{\tr D}. \ee
\item[{\rm(f)}] The polarization formula is valid on $\cH$:
\[
L_{C,D}=\frac{1}{4} \,\left(L_{C+D}-L_{C-D}\right)=\frac{1}{2} \,\left(L_{C+D}-L_{C}-L_D\right).
\]
\item[{\rm(g)}] A change of basis or the preservation of similarity: Let $S$ be nonsingular and
    $D_j=SC_jS^{-1},\, j=1,2$. Then
\[
L_S\,L_{C_1,C_2}\,L_S^{-1}=L_{D_1,D_2}
\]
\end{itemize}
\end{proposition}
\begin{proof}
(a) and (b) are obvious.

(c): The first part follows directly from (b). For the second part, suppose that $\det {C}=0$ (or
$\ker {C}\neq 0$), i.e., ${C}\wek x=0$  for some vector $\wek x\neq 0$. Then $L\wek x\wek
x^*={C}\wek x\wek x^*{C}^* =0$ (i.e., $\ker L_{C}\neq 0$).

(d) follows  from duality
\[
\lr {Y}{{C}X{C}^*}=\lr {Y} {C}{{C}X}=\lr {{C}^*Y {C}}{X}.
\]
The trace formulas \refrm{trL} follow from \refrm{jjpp}.

(f): The polarization formula follows from the definition.

(g): Let $Y=C_1XC_2^*$. That is, $Y=S^{-1}D_1SXS^*D_2^*(S^{-1})^*$. Therefore
$L_S(Y)=SYS^*=D_1(SXS^*)D_2^*=L_{D_1,D_2}(L_SX)$.
\end{proof}
\begin{proposition}[The Uniqueness Theorem]\label{U{C}}
Since $C=0$ and $L_C=0$ simultaneously, assume that ${C}\neq 0$.

Then $L_{C}=L_{D}$ on $\cH$ iff $L_{C}=L_{D}$ on $\cC$ iff  ${C}=z{D}$  for some unit $z\in\BC$.
Hence
\begin{itemize}
\item[{\rm (a)}] The subgroup $\set{L_{C}:\det {C} \mbox{\rm ~ is real}}$ admits the unique
    double representation $L_{C}=L_{-{C}}$.
\item[{\rm (b)}] In particular, $L_C=I$ iff $C=\pm I$.
\end{itemize}
\end{proposition}
\begin{proof}
If $L_C=0$, then for $X=I$ we obtain $CC^*=0$, i.e., $\tr(CC^*)=0$, or $C=0$. Let now $C\neq 0$.
The equivalence follows from \refrm{RR}. Since matrices $X=\wek u\wek v^*$ form a  linearly dense
set in $\cC$, the identity $L_{C}=L_{D}$ on $\cC$ is equivalent to \be\label{uv} {C} \wek u
\,({C}\wek v)^* = D \wek u \,(D\wek v)^*,\quad\wek u\in\BH. \ee Applying the trace, $||{C}\wek
u||=||D\wek u||,\,\wek u\in\BH$. In particular, $\ker {C}=\ker D$. Multiplying on the left by
$({C}\wek u)^*$ in \refrm{uv},
\[
||{C}\wek u||^2\, ({C}\wek v)^*= ({C}\wek u,D\wek u)\, (D\wek v)^*.
\]
So, ${C}\wek v =z(\wek u)\, D\wek v$ for every $\wek v$ and $\wek u$ with ${C}\wek u\neq 0$, where
\[
z(\wek u)=\frac{{(D\wek u,{C}\wek u)}} {||{C}\wek u||^2}.
\]
Multiplying \refrm{uv} by $N\wek v$ from the right, we see that $z(\wek u)=z(\wek v)$ on $(\ker
{C})^c=(\ker D)^c$, i.e., $z$ is a constant of modulus 1.  The form of $z(\wek u)$ on $\ker {C}$ is
irrelevant, so we just adopt $z$. In the last statement and its corollary (b) we just have a real
scalar $z$, so it must be $\pm \,1$.
\end{proof}

Operator or just matrix properties of $C$ and $L_{C}$ may be mutually reflected. However, the
``forward reflection'' from $C$ to $L_{C}$ is significantly easier to formulate and prove than the
``backward reflection'' from an alleged $L_{C}$ back to $C$.

\begin{proposition}{~}\label{jaws props}

\begin{itemize}
\item [{\rm (a)}] Let $\gamma_j\sim \bm\eta_j$ and $\delta_k \sim \bm\zeta_k$ be eigensystems of
    $C$ and $D$, respectively. Then $\gamma_j\overline {\delta_k}\sim \bm\eta_j\bm\zeta_k^*$
    belong to the eigensystem  of $L_{C,D}$. In particular,
\begin{itemize}
\item[{\rm (i)}] $\det L_{C,D}=\det C\cdot\overline {\det D},\,\det L_C=|\det C|^2$;
\item[{\rm (ii)}] the operator $L_{C}$ preserves $\det X$ iff $|\det {C}|=1$.
\item[{\rm (iii)}] If is $C$ is diagonalizable so is $L_C$.
\end{itemize}
\item[{\rm (a$'$)}] Conversely, if $C$ is defective, so is $L_C$.
\item[{\rm (b)}] If $C$ and $D$ are unitary, then $L_{C,D}$ is a rotation. The inverse
    implication is false in general. However, $C$ is unitary iff $L_C$ is a rotation. In general,
    $C$ does not have to be a rotation.
\item[{\rm (c)}] If $C$ is positive then $L_C$ is positive. The inverse implication is false in
    general, Moreover, not every positive $L$ has the jaws form.
\end{itemize}
\end{proposition}
\begin{proof}
(a) and its consequences follow directly from the definition as well as the first statement in (b).

(a'): In view of Proposition 4.1(g), we may choose a basis of dimension $n^2$ for $C$ as we please,
so we choose a canonical Jordan form (cf.\ \cite[XI.\S6]{Lang1987}), consisting of Jordan blocks
spread along the diagonal. A Jordan block is of the form $zI+U$, where the only nonzero elements of
$U$ are $u_{j-1,j}=1$. A matrix $\bmat A&0\\0&B\emat$ is defective iff  $A$ or $B$ is defective.
Therefore, we may assume w.l.o.g. that $C=zI+U$, which makes $L_C$ upper triangular, with numbers
$|z|^2$ on the diagonal. Hence, for $z=0$, $L_C\neq 0$  while for $z\neq 0$, $L_C\neq |z|^2I$.
Therefore, $L_C$ is defective.

(b): Let $C$ have nonzero eigenvalues with at least one non-unit eigenvalue and  let eigenvalues of
$D$ be their reciprocals. This makes a counter-example for the inverse implication in (b).

Suppose that $L_C$ is a rotation, i.e., it is unitary with $\det L_C=1$. Then $L_C\,L_C^*=I$
implies that  $CC^*=I$ by Proposition \ref{note:props L{C}} (d).

(c): If $C$ is positive, so $C=D^2$ for some positive $D$, and then $L_C=L_D^2$. By (a), the
eigenvalues of $L_C$ are of the form $\gamma_j\overline{\gamma_k}$. Let
$\gamma_j=\alpha_j+i\beta_j$. Thus, $L_C$ is positive iff
\[
\alpha_j\beta_k=\alpha_k\beta_j,
\]
which is possible even when all $\beta_j\neq 0$ or when the signs of  $\alpha_j$ are mixed, i.e.,
$C$ could have some negative or even non-real eigenvalues yet $L_C$ would be positive.

Let $L$ have $n=d^2$ positive eigenvalues forming the set $\cal E$. In order to represent $L$ as
$L_C$ with a positive $C$ a tabulation $\gamma_{j}\gamma_{k} $ of $\cal E$ is necessary (and rare).
Further, it is necessary to find real vectors $\bm\eta\_j$ so that each eigenvalue
$\gamma_{j}\gamma_{k}$ would own $\bm\eta_j\bm\eta_k^*$. If these vectors are orthogonal, then we
create the matrix $V$ with $\bm\eta_j$ as columns, and put $C=V^*\Gamma V$ where $\Gamma=\diag
(\gamma_j)$.
\end{proof}

\subsection{The Pauli's coding}\label{s:Pauli}
A vector $\wek x=[\ct{x_0,x_1,x_2,x_3]}\in\R^4$ can be coded as a $2\times 2$ complex Hermitian
matrix,
\[
\wek x\quad\mapsto\quad \lsig\wek x=\bmat x_0+x_3 &x_1-ix_2\\x_1+ix_2& x_0-x_3\emat = \sum_{k} x_{k}\,\sigma_{k}=X,
\]
entailing the corresponding basis of Pauli matrices:
\[
 \sigma_0=I=\bmat 1 &0\\0 &1\emat,\qquad \sigma_1=\bmat 0 &1\\1 &0\emat,\qquad \sigma_2=\bmat 0 &-i\\i &0\emat,\qquad \sigma_3=\bmat 1 &0\\0 &-1\emat,
\]
satisfying \be\label{d Pauli}
\begin{array}{rl}
\mbox{(a) \rule{0pt}{10pt}}& \sigma_0=I,\qquad \sigma_k^2=I,\qquad \sigma_k^*=\sigma_k,\qquad \sigma_j \sigma_k=-\sigma_k \sigma_j,\\
\mbox{(b) \rule{0pt}{10pt}}&\lr{\sigma_j}{\sigma_k}=\delta_j^k,\\
\mbox{(c) \rule{0pt}{10pt}}&\sigma_1\sigma_2=i\,\sigma_3.\\
\end{array}
\ee Therefore, we can encode the vector from a Hermitian matrix:
\[
\wek x= X^\sigma,\quad\mbox{where}\quad x_{k}=\lr {\sigma_{k}}{X}=\frac{1}{2}\,\tr \sigma_k X.
\]
In other words,
\[
\lsig\wek e_{k} = \sigma_{k},\qquad \sigma_k^\sigma=\wek e_k.
\]
In virtue of \refrm{RR}, the isomorphism between real vectors and Hermitian matrices extends to the
isomorphism between complex vectors and complex matrices: \be\label{uiv} \lsig(\wek u+i\,\wek v)\df
\lsig\wek u+i\,\lsig\wek v. \ee

An operator $L\in \cL(\BC^d)$ entails its matrix representation $\Lambda=[l_{jk}]\in \cL(\R^{2n})$
through the isomorphism
\[
 L^\sigma=\lsig\jor, \qquad \mbox{i.e.,}\quad  L(\lsig\wek x)=\lsig(\jor\wek x),\quad \wek x\in \BC^4.
\]
In practice it suffices to assign the basic vectors to the columns of $\jor$, i.e.,  $\wek
e_k\mapsto \jor \wek e_k=\bm l _k$, so the $k$\tss{th} column of $\jor$ will appear as the
Hermitian \be\label{lksj} \lsig\bm l_k =\sum_j l_{jk} \,\sigma_j. \ee
\begin{proposition}{~}
\begin{enumerate}
\item An operator $L_{C}$ admits the matrix representation $\jor=\jor_{C}=[l_{jk}]$ such that
    \be\label{ljk} l_{jk}=\lr{\sigma_j}{{C} \sigma_k {C}^*}. \ee
\item Necessarily, $l_{00}=|\tr C|^2\ge 0$ and $l_{00}=0$ iff ${C}=0$.
\end{enumerate}
\end{proposition}
\begin{proof}
It suffices to employ the basis and the duality and then compute $l_{00}$ from \refrm{ljk}.
\end{proof}

The utility of Pauli matrices is strictly confined to four dimensional spaces.
\begin{proposition}
If the $p$-dimensional complex vector space $\BH\otimes \BH=\cL(\BH)$  ($p=d^2$) admits a
Pauli-like basis $\sigma_k$ satisfying {\rm \refrm {d Pauli} (a) -- (b)}, then  $p=4$, i.e., $d=2$.
\end{proposition}
\begin{proof}
We have \be\label{sss} \sum_{k\ge 0} \sigma_k \sigma_j \sigma_k=\left\{
\begin{array}{cl}
p\,\sigma_0,&\mbox{if $j=0$},\\
(4-p)\,\sigma_j,&\mbox{if $j>0$}.\\
\end{array}\right.
\ee Indeed, let $j=0$. Then
\[
\sum_{k\ge 0} \sigma_k \sigma_0 \sigma_k=\sum_{k\ge 0} \sigma_0=p\,\sigma_0
\]
Let $j>0$. Then
\[
\sum_{k\ge 0} \sigma_k \sigma_j \sigma_k=\sigma_j+\sigma_j^3+\sum_{0<k\neq j} \sigma_k \sigma_j \sigma_k=2\sigma_j-(p-2) \sigma_j=(4-p)\,\sigma_j
\]
Using \refrm{sss}, let us evaluate the action of the following operator on  matrices $C=c_0
\sigma_0+\sum_j c_j \sigma_j=c_0 \sigma_0+C_0$:
\[
P (C)\df \sum_{k\ge 0} \sigma_k \,C\, \sigma_k=
p\,c_0 +(4-p)\,C_0.
\]
 Consider the jaws operator $L(X)=CXC^*$ and let $\jor=\jor_C=[l_{jk}]$ be its matrix representation w.r.t. $(\sigma_k)$, i.e.,
\[
l_{jk}=\lr {\sigma_j}{C \sigma_k C^*}=\frac{1}{d} \,\tr \sigma_j C \sigma_k C^*.
\]
Let us compute its trace:
\[
\begin{array}{l}
 \tr \jor_C =\Sum_j l_{jj}=\sum_j\lr {\sigma_j}{C\sigma_jC^*}=\frac{1}{d}\,
 \tr P(C) C^*=\lr{C}{P(C)}\\
 =\lr {c_0 \sigma_0+C_0}{p\,c_0 \sigma_0+(4-p)\,C_0}=
 |\tr C|^2+(4-p)\,||C_0||^2.
 \end{array}
 \]
 In virtue of Proposition \ref{note:props L{C}}.\refrm{trL}, necessarily $p=4$.
 \end{proof}

Once we have selected and fixed the standard bases as well as the isomorphisms, there is no need to
mark them anymore. That is, instead of writing $\sigma_k=\lsig\wek e_k$ we will simply write
$\sigma_k = \wek e_k$. We will also write $1=I=\sigma_0$. In other words, while coding $X\kor
\zeta+\wek z$, we write in the code:
\[
X=\zeta+\wek z\df\zeta \sigma_0+\sum_{k>0} z_k \wek e_k.
\]
We now stress the typographic distinction between scalars (i.e., scalar multipliers of the identity
operator)  and vectors that allows their quick visual recognition. Accordingly, we denote
$(cz)=\ct{\wek c}\wek z$ even when both vectors are complex. We verify directly that
\[
\wek c\,\wek z=\Big(\sum_{j>0} c_j\wek e_j \Big)\,\Big(\sum_{k>0} z_k\wek e_k \Big)=(cz) +i\,\wek c\cross\wek z.
\]
In particular, for $\wek c=\wek a+i\,\wek b$,
\[
c^2\df\wek c^2=(c\,c)=\ct{\wek c}\wek c=a^2-b^2+2i\,(ab).
\]
 Note that $a^2=||a||^2$ for a real vector $\wek a$, so we may and do assume that $a\ge 0$, whence $a=0$ iff $a^2=0$ iff $\wek a=0$. In contrast, for a non-real complex vector $\wek c$,  $c^2=0$ means that it's real and imaginary components are orthogonal vectors of the same length. However,  the scalar ``$c$'' stays undefined  but we still may write $|c|=||\wek c||$.

Further, unless specifically stated, the presence of a subscript such as in  $\wek e_k$
automatically will mean that $k>0$. According to this convention,   $X^*=(\zeta+\wek
z)^*=\overline{\zeta}+\overline{\wek z}$ and a Hermitian matrix is represented by a real vector
$\zeta+\wek z$. The full multiplication tables emerge as expected: \be\label{full multvec}
(\gamma+\wek c)\,(\zeta+\wek z) = \Big(\gamma\zeta + (cz)\Big)+\Big(\gamma\, \wek z+\zeta\, \wek
c+i\,\wek c\cross\wek z\Big), \ee subject to the tedious but routine split into the real and
imaginary part. That is, letting $\gamma=\alpha+i\beta,\,\zeta=\xi+i\eta$ and $\wek c=\wek
a+i\,\wek b,\,\wek z=\wek x+i\,\wek y$, we list the components:
\[
\begin{array}{ll}
\mbox{scalar:} \rule{20pt}{0pt}& \alpha\xi-\beta\eta+(ax)-(by) +\,i\,\Big(\alpha\eta+\beta\xi+(bx)+(ay)\Big),\\
\mbox{vector:} & \alpha\,\wek x-\beta\,\wek y+\xi\,\wek a-\eta\,\wek b+
-\wek a\cross\wek y-\wek b\cross\wek x\\
&\rule{100pt}{0pt}+\,i\,\Big(\alpha\,\wek y+\beta\,\wek x+\xi\,\wek b+\eta\,\wek a+\wek a\cross\wek x+\wek b\cross\wek y  \Big).
\end{array}
\]
Let us gather a few immediate corollaries.
\begin{proposition} \label{note:inverse}
Let $C=\gamma+\wek c=\alpha+i\beta+ \wek a+i\,\wek b$. Then
\begin{enumerate}
\item $(\gamma+\wek c)\,(\gamma-\wek c)=\gamma^2-c^2$.
\item $\det {C}=\gamma^2-c^2=\gamma^2-\wek c^2 = \alpha^2-\beta^2-a^2 +b^2+2i
    \,\big(\alpha\beta-(ab)\big)$.
\item $\det {C}$ is real iff $\wek a$ and $\wek b$  are $G$-orthogonal.
\item $\det C=1$ iff $\gamma^2-c^2=1$ iff $C^{-1}= \gamma-\wek c$.\cvd
\end{enumerate}
\end{proposition}
Let us write explicitly how $C$ determines the jaws operator $L_C$, i.e., the Lorentz matrix
$\jor$.
\begin{lemma}\label{C2LC}
Given $C= \gamma+\wek c=\alpha+i\,\beta+\wek a+i\,\wek b$,  the operator $L_C$ is represented by a
proper Lorentz matrix
\[
\jor=\bmat s & \ct{\wek q}\\ \wek p & A\emat=
\bmat s & q_1 & q_2 & q_3 \\ \wek p & \wek {a}_1 & \wek {a}_2 & \wek {a}_3\emat
\]
as follows:
\begin{align}
s      & = \alpha^2+\beta^2+a^2+b^2,\label{Ls}\\
\wek p & = 2\,\big(\alpha\,\wek a+\beta\,\wek b+\wek a\cross\wek b\big),\label{Lp}\\
\wek q & = 2\,\big(\alpha\,\wek a+\beta\,\wek b-\wek a\cross\wek b\big),\label{Lq}\\
\wek {a}_j& =(\alpha^2+\beta^2-a^2-b^2)\,\wek e_j +2\,\Big(a_j\,\wek a+b_j\wek b+\wek e_j\cross (\alpha\,\wek b-\beta\,\wek a)\Big),\label{Lkj}\\
 A&= (\alpha^2+\beta^2-a^2-b^2)\,I+2\,\Big(\wek a\,\ct{\wek a}+\wek b\,\ct{\wek b} +\alpha\,V(\wek b)-\beta\,V(\wek a)\,  \Big),\label{LK}
\end{align}
where $V(\wek h)$ denotes the matrix of the cross product operator, $V(\wek h)\,\wek x=\wek x\cross
\wek h$.
\end{lemma}
\begin{proof}
 We will solve the equations that correspond to the action of $L_C$ on the basic vectors,
\begin{align}
(\gamma+\wek c)\,(\cg+\cc)&=s+\wek p,\label{C1}\\
(\gamma+\wek c)\,\wek e_j\,(\cg+\cc)&=q_j+\wek {a}_j.\label{C2}
\end{align}
We calculate
\[
\wek c\,\cc=a^2+b^2+\wek a\cross\wek b\quad\AND\quad
\cg\,\wek c+\gamma\,\cc=2\,(\alpha\,\wek a+\beta\,\wek b),
\]
which entail  the left hand side of \refrm{C1} and, consequently, \refrm{Ls} and \refrm{Lp}:
\[
|\gamma|^2+a^2+b^2+\wek c\,\cc+\cg\,\wek c+\gamma\,\cc=
\alpha^2+\beta^2+a^2+b^2+2\,\wek a\cross\wek b+2\,(\alpha\,\wek a+\beta\,\wek b).
\]
Now, the left hand side of \refrm{C2} equals to \be\label{cejc} |\gamma|^2 \,\wek
e_j+\Big(\gamma\,\wek e_j\, \cc+\cg\,\wek c\,\wek e_j\Big)+\wek c\,\wek e_j\,\cc=|\gamma|^2 \,\wek
e_j+2 \,\Re\,\big(\gamma\,\wek e_j\, \cc\big)+\wek c\,\wek e_j\,\cc. \ee Let us calculate the
portions, denoting the $j$\tss{th} coordinate of the cross product $\wek a\cross \wek b$ by
$(a\cross b)_j$,
\[
\begin{array}{rl}
\wek e_j\,\cc&=a_j+\wek e_j\cross \wek b+i\,(-b_j+\wek e_j\cross\wek a),\\
\Re\,\big(\gamma\,\wek e_j\,\cc\big) &=\alpha\,a_j+\beta\,b_j+\wek e_j\cross (\alpha\,\wek b-\beta\,\wek a),\\
\wek c\,\wek e_j\,\cc =\Re\,\big(\wek c\,\wek e_j\,\cc\big)&
=2\,\Big((a\cross b)_j+a_j\,\wek b+b_j\,\wek a\Big) - (a^2+b^2)\,\wek e_j.
\end{array}
\]
After we substitute these portions first into \refrm{cejc}, and then to \refrm{C2}, the scalar
parts yield \refrm{Lq} and the vector parts yield \refrm{Lkj}. The matrix form \refrm{LK} captures
all of \refrm{Lkj}.\end{proof}

\begin{remark}\label{invcc}
Sometimes it may be worth to transform equations \refrm{C1} and \refrm{C2} by utilizing the inverse
formula $(\zeta +\wek z)^{-1}=\zeta-\wek z$ which is valid when $\zeta^2-z^2=1$,
\begin{align}
\gamma+\wek c&=(s+\wek p)\,(\cg-\cc),\label{invC1}\\
(\gamma+\wek c)\,\wek e_j&=(q_j+\wek {k}_j)\,(\cg-\cc).\label{invC2}
\end{align}
\end{remark}

\subsection{Coding jaws operators}\label{s:jaws}
We underline multiple roles a $2\times 2$ matrix $C$ may play, first as an operator on $\BC^2$,
then as an asymmetric jaws operator $L_{C,I}$ or $L_{I,C}$ acting on $2\times 2$ complex matrices,
and also as the symmetric jaws $L_C$, as well as their combination. These roles might be easily
confused:
\[
C\bm\zeta\, \mbox{ (as an operator on $\BC^2$)}\quad\mbox{vs.}\quad C\,\wek z \,\mbox{ (as an operator on $\BC^4$, formally $L_{C,1}\wek z$)}.
\]
Bringing back the isomorphism mark resolves the issue. While considering the asymmetric jaws
operator $L_{C,I}$ we should write \be\label{Cz} L_{C,I}\wek z = C\lsig\wek z. \ee Yet, for the
sake of clarity we may slightly abuse notation, still writing  ``$C\wek z$'' or ``$\wek z C$''. By
the same token, the formula
\[
C\wek z\wek z^*C^*=L_{C}\wek z\wek z^*
\]
is basically confusion free when the matrix $C$ has been already coded. In contrast, the mark
``$^\sigma$'' might fog the transparency, for the formal script would require
\[
(C\lsig\wek z)(C\lsig\wek z)^*=C\lsig\wek z (\lsig\wek z)^* C^*.
\]
Proposition \ref{jaws props} has indicated that the jaws form of an operator is relatively rare,
and even for such form $L_C$ the transfer of properties backward from $L_C$ to $C$ may be
difficult, even for seemingly simple operators.

\begin{proposition}
$G=\diag (1,-1,-1,-1)$ does not have a jaws form.
\end{proposition}
\begin{proof}
To code $G=L_C$ with $C=\gamma+\wek c$, we would need first
\[
(\gamma+\wek c)(\overline{\gamma}+\overline{\wek c})=1,\quad \mbox{i.e.,}\quad
\overline{\gamma}+\overline{\wek c}=\gamma-\wek c.
\]
So, $\gamma=\alpha$ would be real and $\wek c=i\,\wek b$ would be pure imaginary. Then
\[
(\gamma+\wek c )\wek e_k(\overline{\gamma}+\overline{\wek c})=-\wek e_k,\quad \mbox{i.e.,}\quad
\sigma_k(\overline{\gamma}+\overline{\wek c})=-(\gamma-\wek c )\wek e_k.
\]
In other words, $\wek e_k(\alpha-i\,\wek b)=(-\alpha+i\,\wek b)\,\wek e_k$. Comparing the scalars,
$-b_k=b_k$, i.e. $\wek b=0$. Then $\gamma=\alpha=0$, i.e., $C=0$, a contradiction.
\end{proof}

Nevertheless, the $G$-transpose $\cj{\jor}=G \jor^*G$ of a $4\times 4$ matrix $\jor$, representing
$L_C$ works differently in a simple way.
\begin{corollary}\label{CG}
The $G$-transpose allows the coding $\cj{C}= \cj{(\gamma+\wek c)}\df\gamma-\wek c=C^{-1}$.
\end{corollary}
\begin{proof}
The first equality follows directly from Lemma \ref{C2LC} and the second equality has been stated
in  Proposition \ref{note:inverse}.(e).
\end{proof}

In virtue of the Uniqueness Theorem, Proposition \ref{U{C}}, an ``educated guess'' of the factor
$C$ and the choice of the scalar multiplier may lead to the representation $L=L_C$. Let us focus
now on proper Lorentz matrices. With respect to a fixed basis, a Lorentz matrix admits the unique
polar decomposition $\jor=UP$, where (cf. \refrm{UP}) \be\label{UPnorm}
 U=\bmat 1 & 0\\ 0 & R\emat
\quad\AND\quad P=\bmat {s} &{t}\,\ct{\wek v}\\{t}\, \wek v &S\emat, \ee with the rotation $R$ and
the slider $S=I+({s}-1)\,\wek v\ct{\wek v}$  such that $v=1,\,s\ge 1,\,t^2=s^2-1$.

\begin{theorem}\label{note:P}
Let $P$ and $U$ be given by \refrm{UPnorm}. Then there exist the unique positive $C$ with $\det
C=1$ and the unique rotation $D$ (of course, $\det D =1$) such that $P=L_C$ and $U=L_D$. Therefore,
every proper Lorentz matrix represents a jaws operator, $\jor=L_M$, with $M=DC$ and $\det M=1$. $M$
is unique up to the sign, $\pm M$.
\end{theorem}
\begin{proof}
We can choose either the formulas from Lemma \ref{C2LC} or solve equations \refrm{invC1} and
\refrm{invC2}, as indicated by Remark \ref{invcc}. Let us select the second venue to illustrate
this alternative.

For the rotation, we try $D=\delta+i\,\wek d$ with $\det D=\alpha^2+b^2=1$, $\alpha\ge 0$, and real
$\wek d$.  Denote by $\wek r_j$ the columns of $R$, the direction of the axis by $\wek r$, and  the
angle of rotation by $\rho$. We should have
\[
2+2\cos\rho =\tr U =|\tr D|^2=4 \alpha^2,
\]
whence
\[
\delta=\cos\frac{\rho}{2},\quad b=\sin\frac{\rho}{2}.
\]
Examining the action of $L_D$ on the basis,  $(\delta+i\,\wek d)(\delta-i\,\wek d)=1$ as expected,
and
\[
(\delta+i\,\wek d)\,\wek e_j\,(\delta-i\,\wek d)=\wek r_j.
\]
 Rewriting the equations,
\[
\delta+i\,\wek d\,\wek e_j=\wek r_j\,(\delta+i\,\wek d).
\]
Comparing the scalar parts,
\[
b_j=(r_j b)\quad\imp\quad \wek d=R\wek d,
\]
i.e., $\wek d$ lies on the axis of $R$, $\wek d=\sin\frac{\rho}{2}\,\wek r$.

For $t=0$, i.e., $s=1$, the matrix $P$ reduces to the identity, so we choose $C=I$, obviously. Let
$t>0$, and let us try to find a real $C=\alpha+\wek a$ with $\alpha\ge 0$ and $\det
C=\alpha^2-a^2=1$. Denote the columns of $S$ by $\wek s_k$ and put $\wek c=t\wek v$. Equations
\refrm{lksj} now read:
\[
(\alpha+\wek a)(\alpha+\wek a)=s+\wek c,\qquad
(\alpha+\wek a)\wek e_j(\alpha+\wek a)=c_j+\wek s_j.
\]
The scalar part of the first equation yields $\alpha^2+a^2=s$. The assumption $\det C=1$ gives
\[
\alpha=\sqrt{\frac{s+1}{2}},\qquad a=\sqrt{\frac{s-1}{2}}.
\]
Again, the equations can be rewritten as follows:
\[
\alpha+\wek a=(s+\wek c)(\alpha-\wek a),\qquad
(\alpha+\wek a)\wek e_j=(c_j+\wek s_j)\,(\alpha-\wek a).
\]
Comparing the scalar parts,
\[
\alpha=s\alpha-(ca),\qquad a_j=\alpha c_j-(s_ja).
\]
Substituting  back $\wek c=t\wek v$ and rewriting the equations on the right in the matrix form, we
obtain
\[
\begin{array}{rl}
\vspandexsmall\alpha=s\alpha-t(va)&\imp\quad (va)=\Frac{(s-1)\alpha}{t}=\sqrt{\frac{s-1}{2}},
\\
\vspandexsmall\wek a=\alpha\,\wek c-S\wek a&=\alpha t\, \wek v-(I+({s}-1)\,\wek v\ct{\wek v})\,\wek a\\
&=\alpha t\,\wek v -\wek a-({s}-1)\Frac{(s-1)\alpha}{t}\,\wek v.\\
\end{array}
\]
Simplifying,
\[
\alpha \,\left(t -\frac{(s-1)^2}{t}\right)=\frac{\sqrt{2} s}{\sqrt{s-1}}
\]
Thus
\[
\wek a =\frac{s}{\sqrt{2(s-1})}\,\wek v.
\]
Thus, the search for $D$ and $C$ has been successful.
\end{proof}
\subsection{Exponentials}\label{s:exp}
We will show that the components $C$ (or more precisely, asymmetric jaws $L_{C,I}$) of the jaws
operator $L_C$, represented by a proper Lorentz matrix $\jor$, admit a differentiable
parametrization $C(\zeta)$, with respect to a complex or real variable  $\zeta$, fulfilling the
group property \refrm{group}. Hence, by \ref{expass} $C(\zeta)$ admits a generator $D$, so that
$C(\zeta)=e^{\zeta\,D}$. Therefore, $L_\zeta=L_{C(\zeta)}$ also satisfies \refrm{group} with the
generator $F$. The generators are  related by the formula \be\label{genF} F\,X=DXC_0^*+C_0XD^*\quad
\mbox{($=2 \,\Re\,D XC_0^*$ for a Hermitian $X$)}, \ee which follows from the product rule,
\[
\frac{d L_{\zeta} }{d\zeta}\, X=C_\zeta\rule{.25pt}{0pt}' X C^*_\zeta+C_\zeta X {C^*_\zeta}'.
\]
Let us also issue the warning:
\[
C=e^D \,\,\not{\rule{-10pt}{0pt}\Longrightarrow}\,\, L_C= e^{L_D}\quad\mbox{and}\quad C_1=e^{D_1},\,C_2=e^{D_2} \,\,\not{\rule{-10pt}{0pt}\Longrightarrow}\,\, L_{C_1C_2}= e^{L_{D_1}}e^{L_{D_2}}.
\]
\subsubsection{The diagonalizable matrix}\label{ss:diag}
Recall that in contrast to the real case, $c^2=0$ does not mean that $\wek c=0$. This case will be
handled in the next subsection.

\begin{proposition}\label{exp:diag}
Consider the decomposition  $\gamma+\wek c=(\delta +i\wek d)\, (\alpha+\wek a)$, where $\wek
c^2\neq 0$. Then the following parametrization entails the corresponding exponentials,
\be\label{expnrv}
\begin{array}{lrcl}
\mbox{(general)}\rule{40pt}{0pt}&\gamma+\wek c &=\cosh z+\sinh z\,\wek n
                         &=\expwek{z}{n}\\
\mbox{(rotation)}&\delta +i\,\wek d &= \cos{\theta} +i\sin\theta\, \wek r
                         &=\expwek{i\theta}{r},\\
\vspandexsmall
\mbox{(positive)}&\alpha +\wek  a&= \cosh \phi +\sinh \phi \,\wek v
                         &=\expwek{\phi}{v},\\
&&\expwek{(\phi+i\theta)}{n}=\pm\,\expwek{i\theta}{r}\,\expwek{\phi}{v},&
\end{array}
\ee where the vectors $\wek r$ and $\wek v$ are real, $n^2=r^2=v^2=1$, $\theta=\Frac{\rho}{2}$ and
$\rho$ has been defined defined in the proof of Theorem \ref{note:P}. Note that
\[
(\phi+i\theta)\wek n \neq i\,\theta\,\wek r+ \phi\,\wek v.
\]
\end{proposition}
\begin{proof} It suffices to consider the general case of the first line, for the next lines are just the special cases.
Since $\det M=\gamma^2-m^2=1$, we can choose a nonzero scalar $z\in \BC$ such that
\[
\gamma^2=\cosh^2 z,\qquad m^2=\sinh^2 z,
\]
and normalize the matrix $M$, coded by $\wek c$: \be\label{what z} \wek n \df \frac{\sinh
\overline{z}}{|\sinh z|^2}\,\wek c. \ee Since $\wek n^2=1$, hence
\[
\expwek{z}{n}=\sum_{n=0}^\infty \frac{z^n}{n!}\,\wek n^n=\cosh z\, +\sinh z\,\wek n,
\]
as stated in \refrm{expnrv}. In conclusion, both $C=\expwek{(\phi+i\theta)}{n}$ and
$UP=\expwek{i\theta}{r}\,\expwek{\phi}{v}$ induce the same jaws operator, hence they are exact up
to a sign.
\end{proof}

The surjectivity of the exponential map, based upon diagonalizability, has been established before,
cf. Corollary \ref{expsurj}. However, it was just a nonconstructive  statement of existence. In
contrast, the complex coding yields a generator, or rather ``the generator'',  precisely.
\begin{corollary}
Let $\jor$ be a proper Lorentz matrix, coded as the jaws operator $L_C$ with $C\kor\gamma+\wek c$
with   $c^2\neq 0$ and normalized $\wek n=\wek d+i\,\wek h$.  Then $\jor=e^{F(\mbox{\footnotesize
$\wek d,\wek h$})}$.
\end{corollary}
\begin{proof}
Let us change slightly the parametrization, defining the asymmetric jaws operators
\[
C_z= \cosh\frac{z}{2}+\sinh  \frac{z}{2}\,\wek n,
\]
which satisfy \refrm{expass}.   Computing the derivative we find the symmetric jaws generator,
acting on $\cH$ and extendable onto $\cC$,  by \refrm{expass}. Slightly abusing notation, cf.
\refrm{Cz},
\[
C'_z=\frac{1}{2}\,C_z\wek n \quad\imp \quad FX=\frac{1}{2}\,\Big(\wek n X+X\overline{\wek n}\Big)=
\Re\,\wek n X.
\]
We check its action on the Pauli basis,
\[
F^\sigma 1=\wek d,\qquad F\lsig\wek e_j = d_j+\wek e_j\cross\wek h, \quad \mbox{i.e.,}\quad
F=\bmat 0 & \ct{\wek d}\\ \wek d & V(\wek h)\emat=F(\wek d, \wek h),
\]
which is a $G$-antisymmetric matrix.\end{proof}

\subsubsection{The defective matrix}\label{ss:defective}
While the handling of diagonalizable matrices is relatively simple with the help of powerful tools,
yet, in contrast,  defective matrices  usually  cause trouble. However, in the context of proper
Lorentz matrices, their behavior  pattern is  strikingly simple.

 In this subsection we analyze the eigen-structure of $C=\gamma+\wek c=\alpha+i\,\beta+\wek a+i\,\wek b$, focusing especially on defective matrices.

\begin{proposition}\label{eigeig}{~}
\begin{enumerate}
\item The eigenvectors of $C$ correspond to eigenvectors of the induced asymmetric jaws operator
    $L_{C,I}$ through the
 following relation, \be\label{uu*} C\bm\zeta=\lambda\,\bm\zeta\quad\mbox{iff}\quad
C^\sigma\bm\zeta\bm\zeta^* =\lambda\, \bm\zeta\bm\zeta^* \quad\mbox{iff}\quad (\gamma+\wek
c)\,\bm\zeta\bm\zeta^* =\lambda\,\bm\zeta\bm\zeta^*, \ee which is tantamount to $(\gamma+\wek
c)\,\bm\zeta =\lambda\,\bm\zeta$.
\item The eigenvectors of $C$ correspond to eigenvectors of the induced symmetric jaws operator
    $L_{C}$ through the
 following relation,
\[
C\bm\zeta_j=\lambda_j\,\bm\zeta_j,\quad \imp\quad C^\sigma\bm\zeta_j\bm\zeta_k^*C^* =\lambda_j\overline{\lambda_k}\, \bm\zeta_j\bm\zeta_k^*
\quad\mbox{iff}\quad (\gamma+\wek c)\,\bm\zeta_j\bm\zeta_k^* =\lambda\,\bm\zeta_j\bm\zeta_k^*.
\]
 \item The invertion of the coding $\bm\zeta\bm\zeta^*\,\mapsto\,H$ or
     $\bm\zeta_j\bm\zeta_k^*\,\mapsto\,H+i\,G$ becomes a rather cumbersome task. Fortunately, we
     do not need it.
\item Nonzero scalar factors do not affect eigenvectors, that is, $\gamma^{-1}C\lsig \wek
    u=\lambda\,\wek u$ iff $C\lsig\wek u=\gamma\lambda\,\wek u$. Hence, the technicalities may be
    alleviated by reducing $C$ to the form \be\label{1form} C:=C=1+t\,\wek w=1+t\,(\wek d+i\,\wek
    h),\qquad\mbox{where}\quad d=h=1,\quad t\in\R. \ee
\item For a general $\wek f=\wek g+i\,\wek h$, we have $\wek f\wek f^*=g^2+h^2+2\,\wek
    g\cross\wek h$. Normalizing,
\[
\wek f\wek f^*\,\mapsto\, 1+s\,\wek e,
\]
 where  $\wek e=\wek g\cross\wek h$, $0\le s\le 1$, and
  \[
  \wek f= p\,\wek g+q\,\wek h, \quad \mbox{where   $g=h=1,\, p^2+q^2=1,\,2pq=s$}.
  \]
Thus $(p\pm q)^2=1\pm s$ yields four obvious solutions.
 \item \label{n} Point 4. nd 5. simplify  the eigen-equation:
 \[
\big(1+t\,\wek w\,\big)\,(1+s\,\wek e)=\lambda\,(1+s\,\wek e)\quad\Leftrightarrow\quad
\left\{\begin{array}{l}
1+st\,(we)=\lambda\\
s\,\wek e+t\,\wek w +i\,st\, \wek w\cross\wek e=\lambda s\,\wek e,
\end{array}\right.
\]
 followed by the comparison of the real and imaginary vector parts if needed.

\end{enumerate}
\end{proposition}
\begin{proof} By inspection.\end{proof}

For the remainder of this section we assume that $\wek d\neq 0$ or $\wek h\neq 0$, which excludes
the trivial case  $C=\gamma \,I$. The characteristic polynomial of $C$ is
\[
\det (C-\lambda\,I)= (\gamma-\lambda)^2-c^2, \quad\mbox{where}\quad
c^2=a^2-b^2+2i\,(ab).
\]
There are many ways to detect defective matrices, so let us begin with the double eigenvalue.
\begin{proposition} \label{defective}
Let $\det C=1$. Then the following are equivalent:

\hspace{20pt} {\rm (a)} $c^2=0$;

\hspace{20pt} {\rm (b)} There is a double eigenvalue;

\hspace{20pt} {\rm (c)} $1$ is the double eigenvalue;

\hspace{20pt} {\rm (d)} $\frac{1}{2}\,\tr C=\pm \,1$;

\hspace{20pt} {\rm (e)} $C=\pm \, \big(1+a\,\wek a+i\,a\,\wek b\,\big)$,  with orthogonal unit real
vectors $\wek a, \wek b$.
\end{proposition}

\begin{proof}
If $c^2\neq 0$, then the eigenvalues are $\gamma+\sqrt{c^2}$, accounting for two values of the
complex radical. The double zero occurs iff $c^2=0$ iff $a=b>0$ and $\wek a\perp \wek b$. Also,
$\tr C=2\,\gamma$. Therefore, (a) -- (e) are equivalent.
\end{proof}

Proposition {\rm (d)} above states that the ``angle'' between $C$ and $I$ is either $0$ or $\pi$.
That is, the meaning of ``parallel''  is ambiguous in a matrix space, for it may also mean ``with
the same linear span''.

 Following Proposition \ref{eigeig}.\refrm{n},  we may and do choose the parametrization
\be\label{Cdef} C= C_t\df 1+t\,\wek w,\qquad \wek w\df \wek d+i\,\wek h,\quad t\in\R. \ee Note that
$c^2=0$ in Proposition \ref{defective}.(a) means $w^2=0$. Therefore \be\label{C_comm} (1+t\,\wek
w)\,(1+s\,\wek w)=1+(t+s)\,\wek w. \ee In other words, the asymmetric jaws $L_{C_t,I}$, acting on
Hermitian $2\times 2$ matrices form a commutative operator group. All operators with parameter
$t\neq 0$, besides having the same eigenvalue 1, share the same eigenvectors, which follows
immediately from the decomposition:
\[
C_s =\left(1-\frac{s}{t}\right)\,I +\frac{s}{t}\,C_t.
\]
In fact, the eigen-space is one-dimensional.
\begin{proposition} \label{defective+}
Any condition of Proposition \ref{defective} is equivalent to either of the following:

\hspace{20pt} {\rm (a)} $L_{C,I}$ is defective with $1\,\sim\,1+\wek d\cross\wek h$;

\hspace{20pt} {\rm (b)} $C$ is defective;

\hspace{20pt} {\rm (c)} $C$ represents a shear transformation, i.e., $C$ is similar to $\bmat 1 &
1\\ 0 & 1\emat$.
\end{proposition}
\begin{proof}

We may and do choose $t=1$.  First we show that \refrm{Cdef} implies (a).  Let us  consider  the
eigen-equation
\[
\big(1+\wek w\,\big)\,(1+s\,\wek e)=1+s\,\wek e\quad\Leftrightarrow\quad
\left\{\begin{array}{l}
1+s\,(ue)+i\,s\,(ve)=1\\
\wek u-s\,\wek v\cross\wek e+i\,(\wek v+s\,\wek u\cross\wek e)=s\,\wek e.
\end{array}\right.
\]
Then $s\neq 0$, for otherwise $\wek d=\wek h=0$. So,  $\wek e\perp \wek d,\,\wek e\perp \wek h$,
and thus $\wek e=\wek d\cross \wek h$ together with $s=1$ (from the imaginary vector part) makes
the eigenvector $1+\wek d\cross\wek h$ of $L_{C,1}$, and there is no more eigenvectors independent
of it.

(a) $\Leftrightarrow$ (b): By Proposition \ref{eigeig}, $C$ has two independent eigenvectors iff
$L_{C,I}$ does.

(c) $\imp$ (b) is obvious. (b) $\imp$ (c) (cf.\ the proof of Proposition \ref{jaws props}.(a$'$)):
A defective $C$ with unit determinant is represented by $\bmat 1&z\\0&1\emat$, $z\neq 0$, with
respect to some basis $\wek u,\wek v$, so it suffices to scale $\wek u\mapsto z\wek u$.
\end{proof}

\begin{proposition} Let $C$ have a nonzero vector part and $\det C=1$.

$C$  is defective, i.e.,   $C\kor 1+t\,(\wek d+i\wek h)$ with $\wek d\perp \wek h$ iff  $L_C$ is
defective. Further, $L_C$ has the single eigenvalue $1$ owning the span of $1+\wek d\cross\wek h$
and $\wek h$.
\end{proposition}
\begin{proof} The ``if'' part is implicit in  Proposition \ref{jaws props}.(a).

In addition to $d\cross\wek h$, the eigenvalue 1 also owns $\wek h$, which quickly follows from the
identity
\[
\big(1+t\,(\wek d+i\wek h)\,\big)\,\wek h=\wek v\,\big(1-a\,(\wek d-i\wek h)\,\big).
\]
Should $L_C$ be diagonalizable, so would be  $L^*_C L_C=L_{C^*C}$. Therefore, consider
\[
C^*C= 1+2t^2+2t\,\wek d-2t^2\,\wek d\cross\wek h={\rm const}\cdot (1+\wek e),
\]
where $\wek e\neq 0$. Thus,  it suffices to solve
\[
(1+\wek e)\,(\kappa+\wek {k})=\lambda
(\kappa+\wek {k}) \, (1-\wek e)
\]
for a real $\lambda,\,\kappa$ and $\wek {k}\neq 0$. By comparing the scalar parts, $(wk)=-(wk)$,
i.e. $\wek {k}\perp \wek e$. Thus, there are at most two vectors $\wek {k}$. The comparison of real
vectors yields
\[
\wek {k}+\kappa \,\wek e=\lambda \,(\wek {k} -\kappa  \,\wek e).
 \]
 Crossing it with $\wek e$ yields $\lambda=1$, and crossing it with $\wek {k}$ yields $\kappa=0$.

That is, the eigenspace of $L_C^*L_C$, hence of $L_C$, is two-dimensional.
\end{proof}

Defective shears $C_t=L_{C_t,I}$ with $\det C_t=1$ follow simple operational patterns: For example,
\[
C_t^p=C_{tp}= 1+pt\,\wek w,\quad \,e^{C_t}=e\,C_t=e\,(1+t\,\wek w),\,\mbox{etc.}
\]
Defective jaws $ L_t\df L_{C(t)}$ perform similarly as well since they share the group feature,
because \refrm {C_comm} implies
\[
L_sL_t=L_{s+t}.
\]
\begin{proposition} The scaled parametrization  $C_t=1+\frac{t}{2} \wek w$, where $\wek w=\wek d+i\wek h$, entails the generator  $F=F(\wek d, \wek h)$ of the operator group $L_t$,
\be\label{expFshort} e^{tF}=I+t\,F+ \frac{t^2}{2}\,\tF,\quad \mbox{where $F^2=\tF=L_W$ and $W\kor
\wek w$.} \ee
\end{proposition}
\begin{proof}
We have just proved that the assumptions \refrm{expass} are satisfied.  Hence, \refrm{genF} gives
the generator, which on $\cH$ acts as follows. For a Hermitian $X=\xi+\wek x$, slightly abusing
notation, cf. \refrm{Cz}, we have
\[
FX=\frac{1}{2}\,\Big(\wek w X+X \overline{\wek w}\Big)=
\Re\,\wek w X
\]
Let us compute $F^2$, for a Hermitian $X$:
\[
F^2X=\wek w\,\Big(\wek w\,X+X\, \overline{\wek w}\Big)+\Big(\wek w\,X+X\, \overline{\wek w}\Big)\,\overline{\wek w}=2\,\wek w\,X\,\overline{\wek w}\df\tF\,X.
\]
Since $\wek w^2=0$, then $ F^n=0$ for $n\ge 3$, which gives \refrm{expFshort}.

Alternatively, Lemma \ref{C2LC}  has provided the passage from $C_t$ to $L_t$:
\[
\begin{array}{c}
\vspandexsmall s=1+\Frac{t^2}{4},\,\qquad \wek p=t\,\wek d+\Frac{t^2}{4}\,\wek d\cross\wek h,\qquad
\wek q=t\,\wek d-\Frac{t^2}{4}\,\wek d\cross\wek h\\
A=\left(1-\Frac{t^2}{2}\right)\,I +\Frac{t^2}{2}\,\left(\wek d\ct{\wek d}+\wek d\ct{\wek d}\right)+\,t\,V(\wek h).
\end{array}
\]
Evaluating the derivatives at $t=0$, we obtain the generator of the group $L_t$:
\[
s_0=0,\quad \wek p_0= \wek q_0=\wek d,\quad A_0 = V(\wek h).
\]
We again recognize a $G$-antisymmetric matrix $F=F(\wek d,\wek h)$.
\end{proof}

{\bf Acknowledgement}. The authors are grateful for the referees' commentaries that helped to
improve the appearance of the paper. In particular, Example \ref{ex:1} was suggested by one of the
referees.


\begin{thebibliography}{99}

\bibitem{bartocci2014} C. Bartocci, {Propositions on the structure of the Lorentz group.
    (preliminary and unfinished version)},
    \verb|http://www.dima.unige.it/~bartocci/ifm/gruppo_lorentz.pdf| (Internet Resource), 2014.
\bibitem{barutzenilaufer1994}  A.O. Barut, J.R. Zeni, and A. Laufer, {The exponential map for the
    unitary group SU(2,2)}, {\it J. Phys. A-Math. Gen.} 27:6799-6805, 1994.
\bibitem{dimimlad2005} G. Dimitrov G. and I. Mladenov. A New Formula for the Exponents of the
    Generators of the Lorentz Group. In: {\it Proceedings of the Seventh International Conference
    on Geometry, Integrability and Quantization}, (Eds. I.\ Mladenov, M.\ de Leon), SOFTEX, Sofia
    2005.
    \bibitem{gallier2005} J. Gallier, {Notes on Group Actions Manifolds, Lie Groups and Lie Algebras}, Chap. 18 in: {\it Geometric Methods and Applications}, Springer, Series in Applied Mathematics 38:459-528, 2011.
\bibitem{gottlieb2004} D. H. Gottlieb, {Maxwell's Equations}, Preprint: ArXiv:math-ph /04090012004,
    2004.
    \bibitem{arja-jesz2014} A. Jadczyk and J. Szulga, {A Comment on `On the Rotation Matrix in Minkowski Space-time' by Ozdemir and Erdogdu'}, {\it Rep. Math. Phys.}, 74,1: 39-44, 2014.
\bibitem{Lang1987} S.\ Lang, {\it Linear Algebra (Undergraduate Texts in Mathematics)},
    Springer-Verlag Berlin-Heidelberg, 1987.
\bibitem{Ludyk2013} G. Ludyk, {\it Einstein in Matrix Form: Exact Derivation of the Theory of
    Special and General Relativity without Tensors}, Springer-Verlag Berlin-Heidelberg, 2013.
\bibitem{minguzzi2013} E. Minguzzi, {Relativistic Chasles' Theorem and the Conjugacy Classes of the
    Inhomogeneous Lorentz Group}, {\it J. Math. Phys.} 54:022501, 2013.
\bibitem{naber2012} G. L. Naber, {\it The Geometry of Minkowski Spacetime}, 2nd ed. Springer-Verlag
    New York, 2012.
\bibitem{ozdemir2014} M. \"Ozdemir and M. Erdo\u{g}du, {On the Rotation Matrix in Minkowski
    Space-time}, {\it Rep. Math. Phys.} 74,1: 27-38, 2014.
\bibitem{rademacher1983}    H. Rademacher, {\it Higher Mathematics From An Elementary Point Of
    View}, Birkh\"auser, Basel-Boston-Berlin, 1983.
\bibitem{zenirod1992}        J.R. Zeni and W.A. Rodrigues, {A thoughtful study of Lorentz
    transformations by Clifford algebras}, {\it Int. J. Mod. Phys. A}, 8:1793-1817,  1992.
\end{thebibliography}
\end{document}